\documentclass{article}

\usepackage{microtype}
\usepackage{graphicx}
\usepackage{subfigure}
\usepackage{booktabs} 
\usepackage{float} 
\usepackage{listings}
\usepackage{overpic} 
\usepackage{fancyhdr}

\usepackage{hyperref}

\usepackage[accepted]{icml2025}

\usepackage{amsmath}
\usepackage{amssymb}
\usepackage{mathtools}
\usepackage{amsthm}
\usepackage{multirow}
\usepackage{arydshln}
\usepackage{wrapfig}
 
\usepackage[capitalize,noabbrev]{cleveref}

\theoremstyle{plain}
\newtheorem{theorem}{Theorem}[section]
\newtheorem{proposition}[theorem]{Proposition}
\newtheorem{lemma}[theorem]{Lemma}

\theoremstyle{definition}
\newtheorem{definition}[theorem]{Definition}
\newtheorem{assumption}[theorem]{Assumption}
\theoremstyle{remark}
\newtheorem{remark}[theorem]{Remark}

\renewcommand\algorithmiccomment[1]{%
  \hfill\#\ {\textit{#1}}%
}

\usepackage[textsize=tiny]{todonotes}

\icmltitlerunning{Under review}

\definecolor{codegreen}{rgb}{0,0.6,0}
\definecolor{codegray}{rgb}{0.5,0.5,0.5}
\definecolor{codepurple}{rgb}{0.58,0,0.82}
\definecolor{backcolour}{rgb}{0.95,0.95,0.92}

\lstdefinestyle{mystyle}{
    commentstyle=\color{red!50!green!50!blue!50},  
    keywordstyle=\color{blue!70},  
    numberstyle=\tiny\color{codegray},  
    stringstyle=\color{codepurple},
    basicstyle=\ttfamily\scriptsize,
    breakatwhitespace=false,
    breaklines=true,  
    captionpos=b,
    keepspaces=true,
    numbers=left,  
    numbersep=5pt,
    showspaces=false,
    showstringspaces=false,  
    showtabs=false,
    tabsize=2,
    frame=single  
}
\begin{document}

\twocolumn[
\icmltitle{FeedSign: Robust Full-parameter Federated Fine-tuning of Large Models with Extremely Low Communication Overhead of One Bit}



\icmlsetsymbol{equal}{*}

\begin{icmlauthorlist}
\icmlauthor{Zhijie Cai}{equal,xxx,yyy}
\icmlauthor{Haolong Chen}{equal,xxx,yyy}
\icmlauthor{Guangxu Zhu}{xxx,yyy}
\end{icmlauthorlist}

\icmlaffiliation{yyy}{School of Science and Engineering, The Chinese University of Hong Kong, Shenzhen, Shenzhen, Guangdong, China}
\icmlaffiliation{xxx}{Shenzhen Research Institute of Big Data, Shenzhen, Guangdong, China}

\icmlcorrespondingauthor{Guangxu Zhu}{gxzhu@sribd.cn}


\vskip 0.3in
]



\printAffiliationsAndNotice{\icmlEqualContribution} 

\begin{abstract}

Federated fine-tuning (FFT) attempts to fine-tune a pre-trained model with private data from distributed clients by exchanging models rather than data under the orchestration of a parameter PS (PS). To overcome the bottleneck forged by the growing communication and memory overhead for clients in such systems due to the ever-increasing model sizes, we propose \textit{FeedSign}, an FFT algorithm in which the upload and download payload is exactly $1$ bit per aggregation step. Meanwhile, the memory overhead is squeezed to the amount needed for inference. The reduction is realized by utilizing zeroth-order (ZO) optimizers on large models and shared pseudo-random number generators (PRNG) across devices to represent the gradient estimates as seed-sign pairs. We conduct theoretical analysis on \textit{FeedSign} and show that it converges at an exponential rate $\mathcal{O}(e^{-t})$, where $t$ is the number of elapsed steps under widely used assumptions. Moreover, FeedSign is robust against data heterogeneity and Byzantine attacks. We extensively tested \textit{FeedSign} on models across different structures and sizes (11M to 13B). We found that the proposed method performs better or closely, depending on the scenarios, than its ZO and FO counterparts, albeit with an orders-of-magnitude lower communication overhead. We also discuss some interesting advantages as byproducts guaranteed by the design of \textit{FeedSign}.

\end{abstract}

\section{Introduction}

The development of deep learning (DL) has allowed us to enjoy better intelligent services by training larger models on broader data. While large models demonstrate exemplary performance in general cases, they hold promise to provide more tailored services and greatly improve their ability on intelligent applications if the model can access the rich but privacy-sensitive local data of users. Moreover, since the projected size of datasets used for model pre-training is approaching that of total human data stock \cite{villalobos2024will}, federated learning (FL) \cite{pmlr-v54-mcmahan17a} stands out as a solution to achieve privacy-preserving distributed learning by frequently averaging the model parameters generated by the local data but leaving the data intact at their holders. The algorithm is known as federated averaging (\textit{FedAvg}). When the paradigm is applied to a fine-tuning task, it is often termed federated fine-tuning (FFT) \cite{popov2018distributed}.

Such a learning paradigm demands that stochastic gradient descent (SGD) algorithms be run on client devices. However, the assumed participating client devices are usually resource-restrictive devices like phones and tablets, where SGD can be a heavy computation burden. Moreover, as large models become increasingly popular due to their versatility and great performance, model update aggregation under the FFT paradigm becomes prohibitively expensive. Different methods have been proposed to lower communication and computation costs. Pioneering works include model splitting \cite{thapa2022splitfed} that proposed splitting the DL model into two parts so that most of the computation overhead can be unloaded to the PS as well as reducing communication costs. As for the special case of large models, one of the most successful methods is parameter-efficient fine-tuning (PEFT), which focuses on updating only a small part of the large models and lowering communication and computation costs. Some of the techniques include LoRA \cite{hu2021lora}, Prefix Fine-tuning \cite{li2021prefix}, BitFit \cite{zaken2021bitfit} and Adapter \cite{narayanan2021efficient, pfeiffer2020adapterfusion}.

\begin{figure}[t]
    \centering
    \includegraphics[width=1\linewidth]{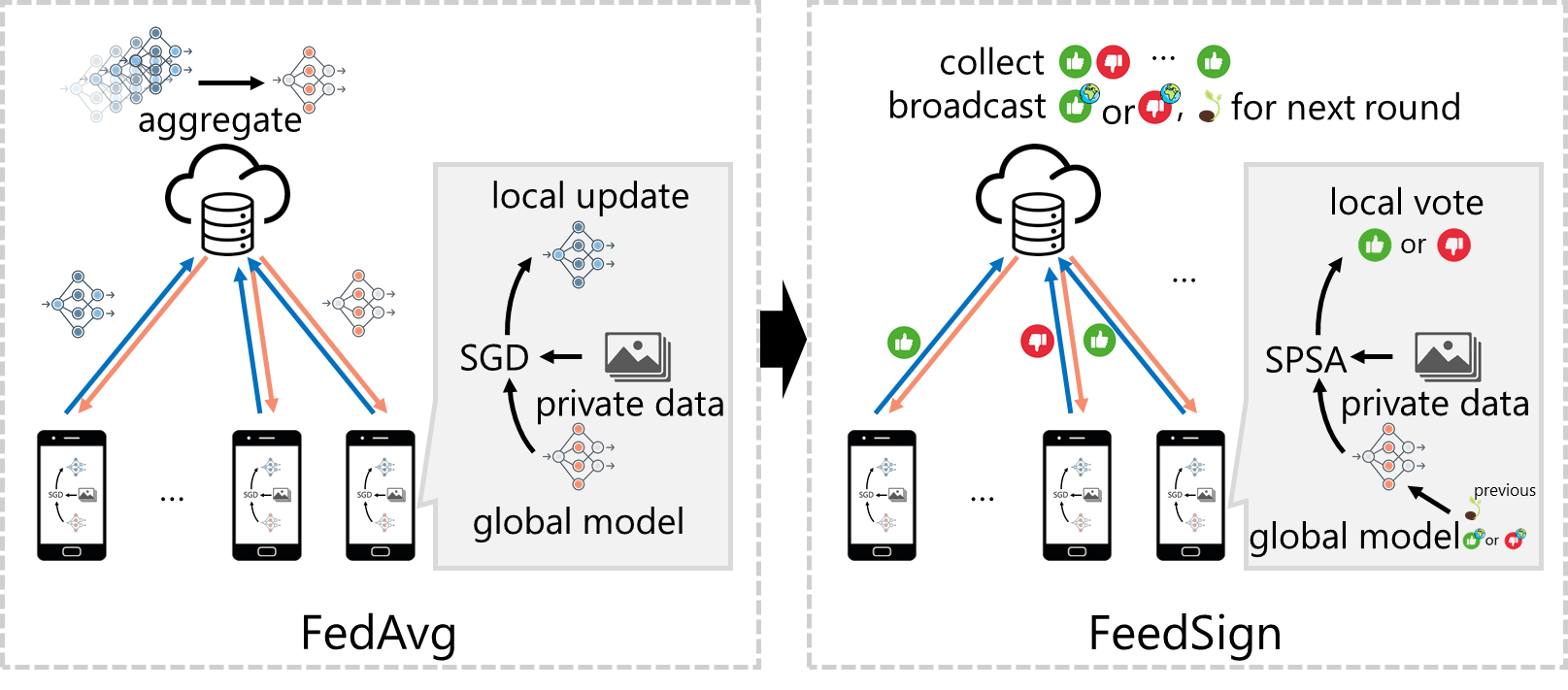}
    \caption{Overview of \textit{FedAvg} and \emph{FeedSign}}
    \label{fig:overview}
\end{figure}

However, while the methods above hold the promise of largely reducing the number of trainable parameters (which scales to the communication cost) with little performance drop, the communication cost of FFT tasks is still formidable. As a qualitative comparison, to participate in FFT on an OPT-1.3B model, a client device will upload around $48$ million float numbers during a communication round, which takes up around $48$ MB, the size of $4$ minutes of YouTube full high definition (FHD, 1080p) video, whereas FFT usually takes hundreds if not thousands of communication rounds to converge, apart from $10$ times of memory demand compared to that of inference\footnote{Specifications on the measurements can be found in Appendix \ref{appd:memory}.}.

A series of pioneering works \cite{xu2024fwdllm, qin2023federated} leverages zeroth-order optimization and the shared Pseudo Random Number Generators (PRNG) across modern deep learning frameworks like PyTorch \cite{paszke2019pytorch} and Tensorflow \cite{abadi2016tensorflow} to lower the per-step uplink communication overhead to acceptable level and the memory demand to an amount around that of inference. However, we show that the per-step uplink and downlink communication overhead can be further reduced to $1$ bit per step regardless of model size with little if not zero performance loss 
\color{black} %
but several advantages, including data heterogeneity, Byzantine resilience, and parameter security. Specifically, our contributions are as follows:
\color{black}

\begin{enumerate}
    \item Establishing upon MeZO \cite{malladi2023fine}, we proposed \emph{FeedSign}, an FFT framework compatible with both full-parameter fine-tuning and PEFT, 
    \color{black} %
    featuring per-step uplink communication overhead of $\boldsymbol{1}$ \textbf{bit} and inference-level memory demand, \textbf{regardless of model size}. \color{black}
    The reduction is realized by utilizing zeroth-order (ZO) optimizers on large models and shared pseudo-random number generators (PRNG) across devices to split the gradient estimates from the clients to \textbf{seed-sign pairs}, each comprises 1) a direction corresponding to a designated random seed and 2) a binary vote from the client indicating whether the seed-corresponding direction grants a local loss descent, which is the only information the clients should convey to the PS.
    \item We provide the convergence analysis of our method and the baselines. We found that it converges at an exponential rate $\mathcal{O}(e^{-t})$, the same rate as in first-order (FO) methods can attain in big $\mathcal{O}$ notation, where $t$ is the number of elapsed steps. The analysis implies that \emph{FeedSign} has additional effects addressing some long-standing problems of FL, including communication bottleneck, data heterogeneity, and Byzantine vulnerability.
    \item We conduct comprehensive experiments across different model types (ResNet, ViT, and OPT) and scales (11M to 13B) to verify the performance of \emph{FeedSign} across various downstream language and vision tasks. It is observed that, 
    \begin{enumerate}
        \item Compared with the conventional FO counterpart, with \textbf{close-to-zero} communication overhead regardless of the model size ($1$ bit versus $24$ GB per step for OPT-13B) and inference-level memory (around $1/12$ for transformer-based models \cite{malladi2023fine}), \emph{FeedSign} achieves comparable test performance;
        \item Compared with federated ZO baselines, with at most $1/64$ of communication overhead, \emph{FeedSign} achieves comparable test performance in general settings while outperforming remarkably under data heterogeneity and Byzantine attacks.
    \end{enumerate}
    \item We discuss some interesting features as byproducts that \emph{FeedSign} will bring to an FL system on parameter security, hardware requirements, and differential privacy.
\end{enumerate}

\section{Related Works}

\subsection{Federated Learning}

Federated learning (FL) contrasts with centralized learning by training a shared model using data from distributed owners without directly sharing the data, thereby preserving data privacy \cite{pmlr-v54-mcmahan17a}. Although centralized learning usually provides an upper bound of performance, FL has its unique advantages as it can access data originally unavailable to centralized learning due to privacy concerns \cite{yang2018applied, hard2018federated, cormode2018privacy}. It is also suitable for uniting siloed raw data without compromising confidentiality as in various fields like healthcare \cite{ogier2022flamby, rieke2020future} and financing \cite{long2020federated}.

However, FL's privacy protection comes at the cost of frequent model parameter exchanges, creating a communication bottleneck \cite{konevcny2016federated, kairouz2021advances}. The success of large pre-trained models in various tasks \cite{liu2019roberta, achiam2023gpt, jiang2023mistral, dosovitskiy2020image} highlights the need to address this bottleneck. Parameter-efficient fine-tuning techniques, which can reduce the number of trainable parameters, show promise when combined with FL to minimize communication overhead \cite{sun2024improving, cho2023heterogeneous, zhang2023fedpetuning, kim2023client}. \emph{SignSGD} \cite{bernstein2018signsgd} proved that an extreme elementwise one-bit quantization on the gradient estimation can largely cut down the communication load while preserving performance. But we notice that the communication overhead inevitably scales to the number of trainable parameters in all of the methods above, and the unmodified version is often considered \textit{impractical} in terms of communication overhead.

\subsection{ZO Optimization for DL and FL}

Over the years, FO methods like SGD and its variants have been the default choice for DL model training \cite{gardner1984learning, amari1993backpropagation, bottou2010large, kingma2014adam, bottou2018optimization}. This method aims to minimize an objective $\mathcal{L}(\boldsymbol{w})$ that characterizes how bad a function $f_{\boldsymbol{w}}$ parameterized by a numerical vector $\boldsymbol{w}$ is mapping from an input space $\mathcal{X}$ to an output space $\mathcal{Y}$ using the chain rule and automatic differentiation \cite{griewank2014automatic, paszke2017automatic} to approximate the derivative of $\mathcal{L}(\boldsymbol{w})$ with respect to $\boldsymbol{w}$. Nonetheless, some objectives of interest are non-differentiable or whose gradients are expensive to compute, calling for alternatives. They are usually known as ZO optimization since they do not require explicit gradient information for objective minimization. The combination of ZO optimization and FL has been a hot research topic in recent years since, in FL settings, clients are usually resource-limited, and ZO can make the estimation of gradients less expensive \cite{fang2022communication, qiu2023zeroth, chen2024fine, ling2024convergence, maritan2024fedzen}. However, the communication bottleneck remains a huge problem for real-world deployment.

Notably, \emph{FwdLLM} \cite{xu2024fwdllm}, \emph{FedKSeed} \cite{qin2023federated} and \cite{zelikman2023just} are the closest works to ours, where the authors discussed a federated fine-tuning framework that exchanges models by exchanging \textbf{seed-projection pairs} to lower the communication overhead of FFT tasks to an acceptable level. \cite{zelikman2023just} proposed an $8$-bit quantization on the projection to further reduce communication overhead. However, they did not provide convergence analyses to justify the performance. A convergence analysis result on distributed zeroth-order optimization-based deep learning originated from \cite{fang2022communication} is invoked from \emph{FedKSeed} \cite{qin2023federated}, but is only dedicated to the ideal case where all clients are honest and hold independent identically distributed data. Moreover, they are both involved in complicated post-processing or repetitive gradient estimation during each step, which could slow down the learning process. In this work, we not only aim to push the method of seed-projection pairs for model exchanging to its limits but also explore more challenging scenarios both theoretically and empirically. 

\subsection{Data Heterogeneity, Byzantine Attacks, and Compression in FL}

Data heterogeneity is a critical concern in federated learning \cite{ye2023heterogeneous}, where each user holds inconsistent shards that do not represent the overall data distribution well, causing divergent updates and undermining training effectiveness. The inconsistency can cause the global model to converge to suboptima with potential performance loss \cite{karimireddy2020scaffold, li2020federated}. Efforts to address this challenge under FO setting are ongoing \cite{qu2022rethinking, fang2023robust, jiang2023fair, chen2024fair}. 

Additionally, Byzantine clients who maliciously alter their data or models can significantly degrade FL performance, necessitating robust FL algorithms \cite{so2020byzantine, tian2022comprehensive}. Within the context of ZO optimization, \emph{DPZero} \cite{zhang2023dpzero} and \emph{CYBER-0} \cite{delgado2024communication} marked an initial attempt to enhance the Byzantine resilience of ZO-based FL. Various aggregation methods have been proposed to improve FO-based FL resilience against such attacks \cite{blanchard2017machine, yin2018byzantine, alistarh2018byzantine, so2020byzantine}. 

There are also methods that stand out as multi-purpose solutions. \cite{allouah2023fixing} explores a joint defense method against data heterogeneity and Byzantine attacks. \cite{bernstein2018signsgd} uses \emph{SignSGD} to cut down communication load and combat against Byzantine attacks, and \cite{liu2019signsgd} proposed a zeroth-order version of \emph{SignSGD} to further reduce the memory overhead in distributed learning systems. \cite{jin2024sign} uses \emph{signSGD} to counter data heterogeneity and Byzantine attacks in distributed learning. However, the communication overhead remains dependent on the dimension of the model and results in prohibitively high communication loads.  Notably, \cite{lang2023compressed, lang2023cpa} introduces a Byzantine resilient compressed aggregation method for FO-based FL systems where the communication overhead is reduced to $1$ bit per step using nested lattice coding with strict privacy guarantees, demonstrating that well-designed lossy compression can induce strong robustness without obviously compromising the performance. 

However, we notice that most efforts addressing this issue separately do \textbf{accurate} gradient estimation followed by \textbf{lossy} compression, leading to potentially unnecessary computational loads, as the compression eventually negates the costly effort of acquiring an accurate gradient estimation in FO-based methods. Motivated by this, we try to find a framework that runs on gradient estimation that is less accurate but attainable and communicable with much lower overheads, with marginal performance loss. The design will mark a difference in rationale between our work and conventional methods addressing data heterogeneity and Byzantine attacks by compression.

\section{\emph{FeedSign}: From Seed-projection Pairs to Seed-sign Pairs}

\label{sec:FeedAvg}

\subsection{Algorithm Design for Seed-sign Pairs}\label{sec:RGE}

\label{subsec:SPSA}

For transformer-based large models, \textit{training} using gradient-based methods usually takes up $12$ times of the memory that is required by \textit{inference} \cite{malladi2023fine}. The excessive demand for memory is due to complex operations of gradient backpropagation \cite{rumelhart1986learning}. One effective method of depriving the extra demand is using backpropagation-free optimizers as is in \emph{FwdLLM} \cite{xu2024fwdllm} and \emph{FedKSeed} \cite{qin2023federated}. The methods proposed by these two works will be referred to as \emph{ZO-FedSGD} for convenience. A brief description of the ZO-based FL is as Algorithm \ref{alg:feedavg}. Proof and descriptions of symbols can be found in the Appendix.

\color{black} %
\begin{definition}[Client Update]\label{def:spsa}
Consider a batch $\mathcal{B}$ from the dataset $\mathcal{D}$, a DL model whose parameter vector is $\boldsymbol{w} \in \mathbb{R}^d$, and a loss function $\mathcal{L}$, the applied ZO gradient estimator SPSA (Simultaneous Perturbation Stochastic Approximation) estimates the gradient as \begin{align}
    p = \frac{1}{n}\sum_{i=1}^n \boldsymbol{z}_i \frac{\mathcal{L}(\boldsymbol{w} + \mu \boldsymbol{z}_i, \mathcal{B}) - \mathcal{L}(\boldsymbol{w} - \mu \boldsymbol{z}_i, \mathcal{B})}{2 \mu},
\end{align}
where $\boldsymbol{z} \sim \mathcal{N}(\boldsymbol{0}, \boldsymbol{I}_d)$ is a Gaussian vector and $\mu$ is the perturbation scale and $p$ is the gradient projection.
\end{definition}

Given Definition \ref{def:spsa}, following normal choices of $n=1$, we apply a different update rule for \emph{FeedSign} elaborated as \begin{equation}
    \begin{aligned}\label{eq:projection}
    (\textbf{\emph{ZO-FedSGD}}) &\quad \hat{\nabla}_{\boldsymbol{w}}\mathcal{L}(\boldsymbol{w}, \mathcal{B}) = p \boldsymbol{z}; \\
    (\textbf{\emph{FeedSign}}) &\quad \hat{\nabla}_{\boldsymbol{w}}\mathcal{L}(\boldsymbol{w}, \mathcal{B}) = \text{Sign}(p) \boldsymbol{z}.
\end{aligned}
\end{equation}

The gradient estimate generated by SPSA can be broken into two parts: the random vector $\boldsymbol{z}$ and its corresponding gradient projection $p$. As a result, the PS needs only the paired seed and projection, or the {\bf seed-projection pairs}. Different from \cite{lang2023compressed, lang2023cpa}, the shared PRNG  directly spawns the random vector $\boldsymbol{z}$ after which the devices scale it by $p$ to perfectly reconstruct the gradient estimation, which is done as follows:


\begin{definition}[Update Aggregation]
    The global model of FL updates with learning rate $\eta$ under the following rule: \begin{align}\label{eq:update}
        (\textit{at client}) \quad \boldsymbol{w} \gets \boldsymbol{w} -  f(p_1, \dots, p_k) \eta\boldsymbol{z},
    \end{align}
    where \begin{equation}\begin{aligned}\label{eq:server_do}
        (\textbf{\emph{ZO-FedSGD}}) &\quad f(p_1, \dots, p_k) = \frac{1}{K}\sum_{k=1}^K p_k; \\ (\textbf{\emph{FeedSign}}) &\quad f(p_1, \dots, p_k) = \text{Sign}\left(\sum_{k=1}^K \frac{p_k}{|p_k|}\right)
    \end{aligned}\end{equation} with $K$ participating clients.
\end{definition}

\begin{algorithm}[t]
\caption{MeZO-based FL}\label{alg:feedavg}
\begin{algorithmic}
   \STATE {\bf Input:} Initialized model parameters $\boldsymbol{w}_0 \in \mathbb{R}^d$, loss function $\mathcal{L}: \mathbb{R}^d \to \mathbb{R}$, step budget $T$, client index set $k \in \mathcal{K} = \{1, \dots, K\}$, collections of client datasets $\{\mathcal{D}_k\}_{k \in \mathcal{K}}$, perturbation scale $\mu$, learning rate $\eta$
   \STATE {\bf Output:} Trained model parameters $\boldsymbol{w}_T$
   \STATE Clients initialize model to $\boldsymbol{w}_{0}$\;
   \FOR {$t = 1, \dots, T$}
   \STATE PS \textbf{broadcasts seed} $s_t$ \COMMENT{only \textit{\textbf{FeedSign}}}
   \FOR {$k = 1, \dots, K$} 
   \STATE \textit{\# Clients do in parallel}
   \STATE Client update local model according to Equation \ref{eq:update} if receives a \textbf{projection broadcast}\;
   \STATE Client sample seed $s_{t,k}$ \COMMENT{only \textit{\textbf{ZO-FedSGD}}}\;
   \STATE Client set PRNG seed to $s_t$ \COMMENT{only \textit{\textbf{FeedSign}}}
   \STATE Client compute $p_k$ according to Equation \ref{eq:projection}\;
   \STATE Client send $p_k$ to PS \;
   \STATE Client send $s_{t,k}$ to PS \COMMENT{only \textit{\textbf{ZO-FedSGD}}}
   \ENDFOR
   \STATE PS collects $p_1, \dots, p_k$ calculate projection $f(p_1, \dots, p_k)$ according to Equation \ref{eq:server_do}\;
   \STATE PS \textbf{broadcasts projection} $f(p_1, \dots, p_k)$\;
   \ENDFOR
\end{algorithmic}
\end{algorithm}

\begin{remark}
    By discarding the amplitude of the gradient estimates, we push the concept of {\bf seed-projection pairs} to {\bf seed-sign pairs}. Different from \emph{ZO-FedSGD}, \emph{FeedSign} always assumes that all clients perturb their model in the same direction for gradient estimation in an aggregation step. Also, \emph{FeedSign} left the sampling of random seeds to the PS since otherwise, the gradient noise would be too large. The PS uses a majority vote to determine whether the model should {\bf march or retreat} a step of fixed size along the designated direction, allowing a $\boldsymbol{1}$\textbf{-bit} per step communication overhead for \emph{FeedSign}.
\end{remark}

As a result, a comparison of communication overhead between \emph{FeedSign} and the baseline \textit{ZO-FedSGD} is as follows, assuming that only one random seed is explored per step.
\begin{equation}
    \begin{aligned}
    \hspace{-10pt}&(\textbf{\textit{ZO-FedSGD}}) \\
    &\quad 1 \times (\underbrace{32}_{\text{float gradient projection}} + \underbrace{32}_{\text{long integer random seed}}) &= 64 \text{ bits},  \\
    \hspace{-10pt}&(\textbf{\textit{FeedSign}}) \\ 
    &\quad 1 \times \underbrace{1}_{\text{float number as gradient projection}} &= 1 \text{ bit}.
\end{aligned}
\end{equation}

\subsection{Convergence Analysis}

Some well-adopted assumptions are needed to facilitate the convergence analysis.

\begin{assumption}[$L$-smooth, \cite{bottou2018optimization}]\label{ass:l-smooth}
    There is a positive constant $L$ for $\mathcal{L}(\boldsymbol{w})$ that satisfies \begin{equation}\begin{aligned}
        \mathcal{L}(\boldsymbol{w}_{t+1}) \leq \mathcal{L}(\boldsymbol{w}_t) &+ \langle\nabla\mathcal{L}(\boldsymbol{w}_t), \boldsymbol{w}_{t+1} - \boldsymbol{w}_t\rangle \\
        &+ \frac{L}{2} \|\boldsymbol{w}_{t+1} - \boldsymbol{w}_t\|_2^2.
    \end{aligned}\end{equation}
\end{assumption}

\Cref{ass:l-smooth} guarantees that the gradient of the loss
function would not change arbitrarily quickly concerning a
varying model parameter. It is notably an essential assumption for convergence analysis of gradient descent methods.

\begin{assumption}[Local $r$-Effective Rank, \cite{malladi2023fine}, Assumption 1]\label{ass:r-effective}
    There is a matrix $\boldsymbol{H}(\boldsymbol{w}_t) \preceq L \boldsymbol{I}_d$ such that with $G(\boldsymbol{w}_t) = \max_{(\boldsymbol{x}, \boldsymbol{y}) \in \mathcal{D}} \|\nabla \mathcal{L}(\boldsymbol{w}_t, (\boldsymbol{x}, \boldsymbol{y}))\|_2$, \begin{enumerate}
        \item For all $\boldsymbol{w}$ such that $\|\boldsymbol{w} - \boldsymbol{w}_t\| \leq \eta d G(\boldsymbol{w}_t)$, we have $\nabla^2 \mathcal{L}(\boldsymbol{w}) \preceq \boldsymbol{H}(\boldsymbol{w}_t)$.

        \item The effective rank of $\boldsymbol{H}(\boldsymbol{w}_t)$, \textit{i.e.}, $\text{tr}(\boldsymbol{H}(\boldsymbol{w})) / \|\boldsymbol{H}(\boldsymbol{w}_t)\|_{\text{op}}$, is at most $r$.
    \end{enumerate}
\end{assumption}

\Cref{ass:r-effective} is key to the following analysis by indicating
that the bulk of the spectrum concentrates around 0 with only
a small number of outliers. In other words, we assume that
there is a very low dimensional but critical model parameter
subspace in fine-tuning large models.

\begin{assumption}[Unbiased Gradient Estimator with Bounded Data Heterogeneity]\label{ass:unbiased}
    The gradient estimator in Definition \ref{def:spsa} is unbiased, specifically, \begin{align}
        &\mathbb{E}_{\mathcal{B}}[\hat{\nabla} \mathcal{L}_k(\boldsymbol{w}, \mathcal{B})] = \nabla \mathcal{L}_k(\boldsymbol{w}), \\
        &\mathbb{E}_{\mathcal{B}}\left[\|\hat{\nabla} \mathcal{L}_k(\boldsymbol{w}, \mathcal{B})\|_2^2\right] \leq c_g \|\nabla \mathcal{L}_k(\boldsymbol{w})\|_2^2 + \frac{\sigma_g^2}{KB} \mathbb{V}[\nabla \mathcal{L}(\boldsymbol{w})], \\
        &\mathbb{E}_k\left[\|\nabla \mathcal{L}_k(\boldsymbol{w}) - \nabla \mathcal{L}(\boldsymbol{w})\|_2^2\right]\leq c_h \|\nabla \mathcal{L}(\boldsymbol{w})\|_2^2 + \sigma_h^2.
    \end{align}
\end{assumption}

\Cref{ass:unbiased} characterizes how batch gradient estimation
is impacted by batch sampling and data heterogeneity. Setting
$c_g = 1$ and $\sigma_g = 0$ leads to an ideal case where batch gradients
are identical to the true client gradient, and setting $c_h = 0$ and
$\sigma_h = 0$ reduces the problem to an i.i.d. setting.

\begin{assumption}[Polyak-\L ojaciewicz Inequality]\label{ass:PL}
    Assume $\mathcal{L}^* := \min_{\boldsymbol{w} \in \mathbb{R}^d} \mathcal{L}(\boldsymbol{w}) > -\infty$, then there is a constant $\delta > 0$ such that for any $\boldsymbol{w} \in \mathbb{R}^d$, $\mathcal{L}(\boldsymbol{w})$ satisfies \begin{equation}\begin{aligned}
        \|\nabla \mathcal{L}(\boldsymbol{w})\|_2^2 &\geq 2\delta (\mathcal{L}(\boldsymbol{w}) - \mathcal{L}^*).
    \end{aligned}\end{equation}
\end{assumption}

\Cref{ass:PL} is a more general assumption than strong convexity \cite{polyak1964gradient, karimi2016linear} but ensuring the loss landscape is amenable to analysis. \emph{MeZO} \cite{malladi2023fine} added another assumption \begin{align}
    \mathbb{V}[\nabla \mathcal{L} (\boldsymbol{w})] \leq 2 \alpha (\mathcal{L}(\boldsymbol{w}) - \mathcal{L}^*), 
\end{align} for a constant $\alpha > 0$ on gradient variance so that the impact on the optimization trajectory of noisy gradient estimations is controllable by a scaled optimization gap.

\begin{assumption}[Sign-Reversing Probability]\label{ass:srp}
    The gradient projection estimate has a reversed sign compared to the true gradient projection with probability upper bounded by $p_{t,e}$. Specifically, given any $\boldsymbol{z}$, it satisfies\begin{align}
        \text{Prob}_{\mathcal{B}}\left[(\boldsymbol{z}^\top \nabla \mathcal{L}(\boldsymbol{w}, \mathcal{B}))(\boldsymbol{z}^\top \nabla \mathcal{L}(\boldsymbol{w})) < 0\right] \leq p_{t, e}.
    \end{align}
\end{assumption}

\Cref{ass:srp} is a loose assumption in terms of batch gradient projection estimation that does not consider the magnitude of projection but only its sign. We conclude that $p_{t,e} < 1/2$ under a mild assumption as discussed in the Appendix. Note that $p_{t,e}$ is data heterogeneity-independent, implying \emph{FeedSign}'s data heterogeneity robustness. If there were no Byzantine client, the overall sign-reversing probability $p_t = p_{t,e}$.

With the assumptions, the following lemma characterizes the stepwise loss descent bound for \emph{ZO-FedSGD}.

\begin{lemma}[Dimension-free Descent Lemma for ZO-FedSGD, \cite{malladi2023fine}] Given $\mathcal{L}(\boldsymbol{w})$ being a $L$-smooth function and $\hat{\nabla} \mathcal{L}(\boldsymbol{w}, \mathcal{B})$ an unbiased gradient estimator with $\mu \to 0$, the expected per-step loss descent can be bounded as follows:
    \begin{align}
        \mathbb{E}\left[\mathcal{L}(\boldsymbol{w}_{t+1})\right] \leq &\mathcal{L}(\boldsymbol{w}_t) - \eta \|\nabla \mathcal{L}(\boldsymbol{w}_t)\|_2^2  \notag \\
        &+ \frac{L \zeta \eta^2}{2} \mathbb{E}_{k, \mathcal{B}}\left[\|\nabla\mathcal{L}_k(\boldsymbol{w}_t, \mathcal{B})\|_2^2\right].
    \end{align}
    where \begin{align}
        \zeta = \frac{dr + d - 2}{n (d + 2)} + 1
    \end{align} characterize the low-rank effect of the \textbf{gradient estimator}.
    \label{lm:descent}
\end{lemma}

\begin{remark}
    Lemma \ref{lm:descent} is the premise of successful ZO-based fine-tuning of large models. It can be observed that there is only an additional term in the quadratic term compared to that of the FO. It is the previous sense that SPSA-like algorithms result in a $\mathcal{O}(d)$ times larger gradient variance compared to FO methods \cite{nemirovskij1983problem, spall1992multivariate, jamieson2012query, oktay2020randomized}, prohibiting successful training of large models. However, Lemma \ref{lm:descent} refined the bound and found that the gradient variance can be controlled by $\mathcal{O}(r)$, where $r$ is a loss landscape-related parameter known as local effective rank. The parameter $r$ is found to be small compared to the model size $d$ in well-trained DL models 
    \color{black} %
    as reported in \cite{papyan2020traces, ghorbani2019investigation, yao2020pyhessian, sagun2017empirical, wu2020dissecting}.\color{black}
\end{remark}

We present a theorem characterizing the convergence of \emph{FeedSign} with two baselines for comparisons. 

\begin{theorem}[Global Convergence for \emph{FedSGD}, \emph{ZO-FedSGD}, and \emph{FeedSign}] Given all assumption including Assumptions \ref{ass:l-smooth}-\ref{ass:srp} satisfied, with corresponding conditions met, after
    \begin{align}
        t = A \log \frac{\mathcal{L}(\boldsymbol{w}_0) - \mathcal{L}^* - \tilde{C}}{\epsilon}
    \end{align}
    steps, we will have the gap between the expected loss $\mathbb{E}[\mathcal{L}(\boldsymbol{w}_t)]$ and its lowest obtainable loss value $\mathcal{L}^* + \tilde{C}$ smaller than $\epsilon$ with \begin{align}
        (\textbf{\emph{FedSGD}}) & \hspace{3pt} A = 2 \delta \eta - L \delta \eta^2 c_g (1 + c_h) - \frac{L \alpha \sigma_g^2 \eta^2}{KB}, \notag \\
        & \hspace{3pt} C = \frac{L c_g \sigma_h^2 \eta^2}{2}; \\
        (\textbf{\emph{ZO-FedSGD}}) & \hspace{3pt} A = 2 \delta \eta - L \zeta \delta \eta^2 c_g (1 + c_h) - \frac{L \zeta \alpha \sigma_g^2 \eta^2}{KB}, \notag \\
        & \hspace{3pt} C = \frac{L \zeta c_g \sigma_h^2 \eta^2}{2}; \\
        (\textbf{\emph{FeedSign}}) & \hspace{3pt} A = 2\sqrt{\frac{2}{\pi}}\delta\eta^2(1 - 2 \max_{t} p_t), \notag \\
        & \hspace{3pt} C = \frac{L r \eta^2}{2}, 
    \end{align}
    where $\tilde{C} = C / A$ captures the error floor with $0 < A < 1$ and $C > 0$, and $\zeta$ is a factor related to the low effective rank property of the pre-trained model.
    \label{thm:1}
\end{theorem}

\Cref{thm:1} offers several interesting observations on the behavior of \emph{FeedSign} and \emph{ZO-FedSGD}.

\begin{remark}
    \textbf{Convergence Rate Comparison.}\label{rmk:crc}
    \Cref{thm:1} above shows that under a \emph{FedSGD}-style setting, both \emph{ZO-FedSGD} and \emph{FeedSign} converges at an exponential rate $\mathcal{O}(e^{-t})$, the same rate as in FO methods can attain in big $\mathcal{O}$ notation. 
    \color{black} %
    Notably, \emph{ZO-FedSGD} differs from \emph{FedSGD} to only a term characterizing the low-rank property of the pre-trained model $\zeta \sim \mathcal{O}(r)$. \color{black}
\end{remark}

\color{black} %
\begin{remark}
    \textbf{Data Heterogeneity Resilience.} 
    It is observed that the error floor of \emph{ZO-FedSGD} scales to the data heterogeneity parameters $c_g$ and $\sigma_h$ while that of \emph{FeedSign} is independent of them. As a result, under an ideal iid case, the error floor vanishes with $\sigma_h = 0$ and $c_g \ll \infty$, but grows under high data heterogeneity. Contrarily, the error floor of \emph{FeedSign} is connected to the learning rate, which can be suppressed by choosing a sufficiently small learning rate. In summary, we trade for more resilience against data heterogeneity at the cost of having a fixed but small error floor in \emph{FeedSign}.
\end{remark}
\color{black}

\begin{remark}
    \textbf{Byzantine Resilience.} 
    Nevertheless, $p_t$ is a key factor influencing the performance of \emph{FeedSign}. It is noticed that for \emph{ZO-FedSGD} and \emph{FeedSign}, any attacks altering gradient estimation boils down to altering the gradient projection due to the deterministic nature of PRNG. While in \emph{ZO-FedSGD}, clients have some degree of freedom to enact their strategies of attack, hence being more unpredictable, 
    \color{black} %
    the most effective method of damaging convergence of FFT due to the binary voting scheme in \emph{FeedSign} is to always send a reversed sign to PS. 
    A succinct analytic characterization to its impact can be found in the Appendix.
\end{remark}
\color{black}

\section{Experiments}

To validate the effectiveness of the proposed approach, we conducted extensive experiments across different tasks, data heterogeneity levels, and models of various types and sizes. 

\begin{table}[h]
\caption{Brief descriptions on the baselines and \emph{FeedSign}.}
\label{table:desc}
\centering
\begin{tabular}{lp{40pt}p{40pt}p{40pt}}
\hline
Method                            & Forward only?             & Federated?                & Stepwise comm. load \\ \hline
FO                                & $\times$     & $\times$     & N/A                 \\
MeZO                              & $\checkmark$ & $\times$     & N/A                 \\
\emph{ZO-FedSGD} & $\checkmark$ & $\checkmark$ & $64$ bit            \\
\emph{FeedSign}  & $\checkmark$ & $\checkmark$ & $1$ bit             \\ \hline
\end{tabular}
\end{table}

\textbf{Baselines.} To ensure consistency with previous research, we evaluate the same models used in \emph{MeZO}. As shown in \Cref{table:desc}, we compare our method with standard FO methods (use backpropagation, takes up at least $6$ times of memory), centralized ZO method \emph{MeZO} \cite{malladi2023fine} and \emph{ZO-FedSGD} \cite{xu2024fwdllm, qin2023federated}. 
We kept the number of total perturbations consistent with that adopted in \emph{MeZO}. As a result, the number of elapsed steps of \emph{MeZO} is $K$ times that of \emph{ZO-FedSGD} and \emph{FeedSign}.
We run both algorithms for the same number of steps, so the total communication overhead of \emph{FeedSign} is $1 / 64$ of that of \emph{ZO-FedSGD}. We make $5$ repetitive runs with different seeds series and append the standard deviation value in brackets.

\begin{table*}[t]
\setlength{\tabcolsep}{2pt}
\centering
    \caption{Main results on OPT-13B over language tasks. The highest metric obtained using federated ZO optimization is \textbf{bolded}, and the metric gap to that of the FO method is reported in the rightmost column.}
    \label{table:nlp_opt_results}
    \centering
    \begin{center}
        \begin{tabular}{lcccccccccccc}
            \toprule
            \multicolumn{1}{c}{Task}  &\multicolumn{1}{c}{\bf SST-2}  &\multicolumn{1}{c}{\bf RTE}  &\multicolumn{1}{c}{\bf CB}  &\multicolumn{1}{c}{\bf BoolQ} &\multicolumn{1}{c}{\bf WSC}  &\multicolumn{1}{c}{\bf WIC} &\multicolumn{1}{c}{\bf MultiRC} &\multicolumn{1}{c}{\bf COPA}  &\multicolumn{1}{c}{\bf ReCoRD} &\multicolumn{1}{c}{\bf SQuAD}  &\multicolumn{1}{c}{\bf DROP} & \multirow{2}*{Gap} \\ 
            \multicolumn{1}{c}{Type} & \multicolumn{7}{c}{----------------------- classification -----------------------} & \multicolumn{2}{c}{-- multiple choice --} & \multicolumn{2}{c}{--- generation ---} \\
            \hline
            Zero-shot & 58.8 & 59.6 & 46.4 & 59.0 &	38.5 & 55.0	& 46.9 & 80.0& 81.2& 46.2 & 14.6 & -- \\
            \hline
            FO & 92.0 & 70.8 & 83.9 & 77.1 & 63.5 &	70.1 & 71.1 &79.0 & 74.1 & 84.9 & 31.3 & -- \\
            \hdashline
            MeZO & 91.4 & 66.1 & 67.9 & 67.6 & 63.5 & 61.1 & 60.1 & 88.0 & 81.7 & 84.7 & 30.9 & -3.1 \\
            \hdashline
            \emph{ZO-FedSGD} & 85.9&	60.5&	\bf 68.1&	65.4&	52.8&	54.1&	56.1&	85.0&	80.4&	76.3&	\bf 28.9&	-7.6 \\
            & (0.7)&	(1.7)&	(2.6)&	(2.0)&	(9.7)&	(1.2)&	(2.6)&	(1.0)&	(1.0)&	(0.5)&	(0.5) & -- \\
            \bf \emph{FeedSign} & \bf 87.3&	\bf 61.0&	67.8&	\bf 66.6&	\bf 55.5&	\bf 56.0&	\bf 57.4&	\bf 87.8&	\bf 81.3&	\bf 78.0&	28.4&	\bf -6.4 \\
            & (0.5)&	(1.5)&	(0.0)&	(1.5)&	(10.6)&	(1.0)&	(0.7)&	(0.4)&	(0.9)&	(0.7)&	(0.3) & -- \\
            \bottomrule
        \end{tabular}
    \end{center}
\end{table*}

\subsection{Main Results in General Settings}

\textbf{Language models.} We report the results on OPT-13B in \Cref{table:nlp_opt_results}. It can be observed that the mean metric gap from centralized FO to \emph{FeedSign} is $-6.1\%$ across $11$ different tasks. Moreover, despite discarding the amplitude information, the test metrics are higher for most instances compared to \emph{ZO-FedSGD}. The improvement could be attributed to the overfitting prevention effect of featuring a noisy gradient estimator \cite{welling2011bayesian, bernstein2018signsgd}.


\begin{table}[t!]
\centering
\caption{Results on ViT-large FFT, client pool size $K=5$.}
\begin{tabular}{@{}lcc@{}}
\toprule
Dataset   & \bf CIFAR-10 & \bf CIFAR-100 \\ \midrule
ZO-trained SOTA & 86.5 & 34.2 \\ \midrule
\emph{FeedSign}  & \bf 91.9 (5.9)    & \bf 45.3 (5.0)     \\ \bottomrule
\end{tabular}\label{table:vision models}
\vspace{-10pt}
\end{table}

\textbf{Vision models.} \Cref{table:vision models} reports the test accuracy of \emph{ZO-FedSGD} and \emph{FeedSign} on CIFAR-10 and CIFAR-100. We download a pre-trained model checkpoint, replace the classifier layer with a randomly initialized layer, and fine-tune only the last layer. It is shown that \emph{FeedSign} attains a test accuracy of $91.9 \%$ in only $2 \times 10 ^ 4$ steps with the support of a pre-trained model, higher than the ZO from-the-scratch training SOTA \cite{chen2023deepzero, zhang2024foresight} to the best of our knowledge with a much lesser number of steps. 

\subsection{Data Heterogeneity Resilience}\label{chap:exp_heterogeneity}

\textbf{Settings.} A common approach to generating heterogeneous shards of a dataset is to have the number of samples from a class $c$ being proportional to $p_c \sim \text{Dirichlet}(\beta)$ for a client where $\beta$ is a controlling parameter
\cite{vahidian2023rethinking}. 

\textbf{Language models.} \Cref{table:non_iid} reports the test metric of \emph{ZO-FedSGD} and \emph{FeedSign}. We observe a drop in test metrics across all tasks, confirming FL's vulnerability to data heterogeneity. However, it is clear that \emph{FeedSign} outperforms \emph{ZO-FedSGD} on most of the entries.

\textbf{Vision models.} We conduct FFT on a ResNet-18 checkpoint. 
We observe that although \emph{ZO-FedSGD} outperforms \emph{FeedSign} on iid data, \emph{FeedSign} turns the tide under high data heterogeneity, which affirms the theoretically implied data heterogeneity robustness of \emph{FeedSign}.

\subsection{Byzantine Resilience}

\begin{table*}[t]
    \caption{Main results on OPT-125M over language models with non-iid data. We \textbf{bolded} the higher result within \emph{FeedSign} and \emph{ZO-FedSGD}.}
    \label{table:non_iid}
    \begin{center}
        \begin{tabular}{lccccccc}
            \toprule
            \multicolumn{1}{c}{Task}  &\multicolumn{1}{c}{\bf SST-2}  &\multicolumn{1}{c}{\bf RTE}  &\multicolumn{1}{c}{\bf CB}  &\multicolumn{1}{c}{\bf BoolQ} &\multicolumn{1}{c}{\bf WSC}  &\multicolumn{1}{c}{\bf WIC} &\multicolumn{1}{c}{\bf MultiRC} \\
            \hline
            Zero-shot & 51.2 & 53.0 & 48.2 & 41.5 & 37.5 & 51.2 & 49.7 \\
            \hline
             
            MeZO & 82.6 & 53.1 & 67.4 & 62.7 & 61.7 & 54.4 & 55.8 \\
             & (0.7) &	(2.3)&	(0.8)& (0.7) & (3.4) &	(1.3)&	(1.7) \\
            \hline

            \emph{ZO-FedSGD}, $\beta=1.0$ & 80.6&	51.8&	\bf 65.3&	61.4&	56.8&	\bf 52.1&	54.3 \\
             & (3.3)&	(2.1)&	(2.9)&	(1.1)&	(10.8)&	(1.8)&	(5.1) \\
            \bf \emph{FeedSign}, $\beta=1.0$ & \bf 82.6&	\bf 52.8&	\bf 65.3&	\bf 62.5&	\bf 61.1&	\bf 52.1&	\bf 54.7 \\
             & (1.9)&	(3.0)&	(4.6)&	(1.1)&	(4.7)&	(1.8)&	(5.3) \\
            \bottomrule
        \end{tabular}
    \end{center}
\end{table*}

\begin{table*}[h]
\setlength{\tabcolsep}{2pt}
\centering
    \caption{Main results on OPT-125M over language models with a Byzantine attacker.}
    \label{table:byzantine_robustness}
    \begin{center}
        \begin{tabular}{lccccccccccc}
            \toprule
            \multicolumn{1}{c}{Task}  &\multicolumn{1}{c}{\bf SST-2}  &\multicolumn{1}{c}{\bf RTE}  &\multicolumn{1}{c}{\bf CB}  &\multicolumn{1}{c}{\bf BoolQ} &\multicolumn{1}{c}{\bf WSC}  &\multicolumn{1}{c}{\bf WIC} &\multicolumn{1}{c}{\bf MultiRC} &\multicolumn{1}{c}{\bf COPA}  &\multicolumn{1}{c}{\bf ReCoRD} &\multicolumn{1}{c}{\bf SQuAD}  &\multicolumn{1}{c}{\bf DROP}\\ 
            \multicolumn{1}{c}{Type} & \multicolumn{7}{c}{----------------------- classification -----------------------} & \multicolumn{2}{c}{-- multiple choice --} & \multicolumn{2}{c}{--- generation ---} \\
            \hline
            Zero-shot & 51.2 & 53.0 & 48.2 & 41.5 & 37.5 & 51.2 & 49.7 & 69.0 & 51.7 & 9.5 & 4.4 \\
            \hline
            \emph{ZO-FedSGD} & 79.0&	52.1&	66.3&	\bf 59.8&	48.8&	52.2&	53.9&	64.8&	48.0&	35.8&	12.5 \\
            & (1.7)&	(3.2)&	(3.1)&	(2.5)&	(10.2)&	(0.6)&	(1.5)&	(4.1)&	(2.7)&	(0.9)&	(1.0) \\
            \hdashline
            \bf \emph{FeedSign} & \bf 83.2&	\bf 55.9&	\bf 66.7&	59.7&	\bf 50.5&	\bf 52.8&	\bf 55.5&	\bf 66.4&	\bf 49.7&	\bf 42.8&	\bf 14.2 \\
            & (0.5)&	(1.2)&	(2.6)&	(1.1)&	(7.0)&	(1.2)&	(1.7)&	(1.5)&	(2.1)&	(0.6)&	(0.8) \\
            \bottomrule
        \end{tabular}
    \end{center}
\end{table*}

\begin{figure}[t]
\centering
    \includegraphics[width=0.35\textwidth]{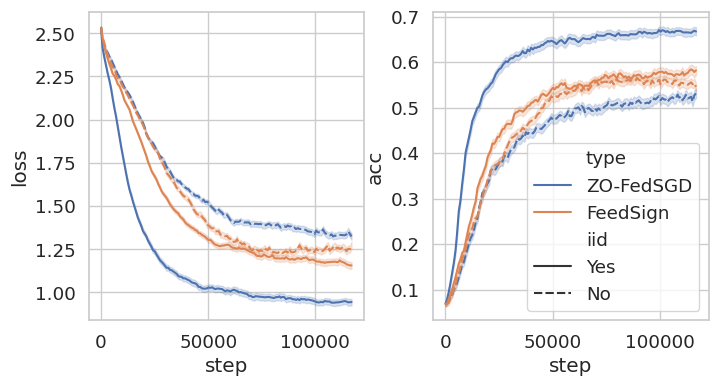}
    \vspace{-10pt}
    \caption{Loss and accuracy curve versus the number of steps elapsed under data heterogeneity.\label{fig:noniid}}
    \vspace{-10pt}
\end{figure}

\begin{figure}[t]
\centering
    \includegraphics[width=0.5\textwidth]{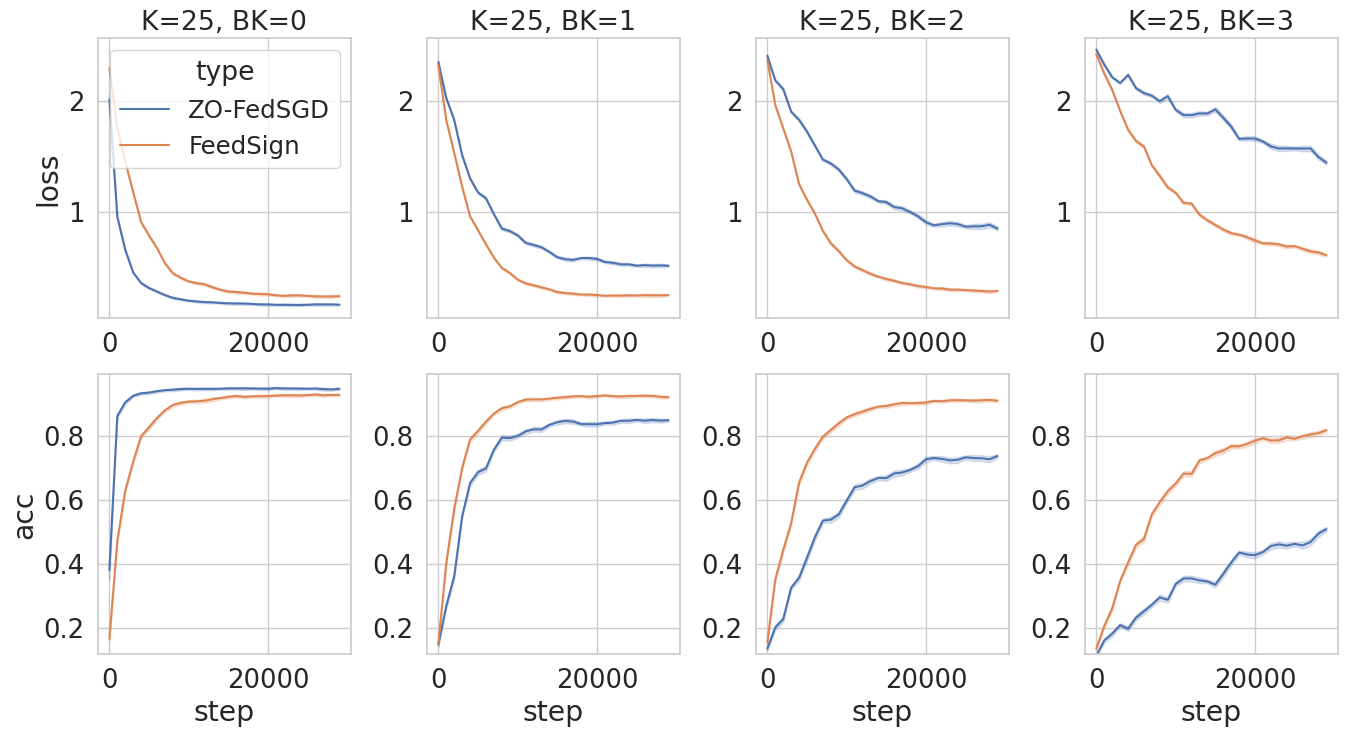}
    \vspace{-10pt}
    \caption{Loss and accuracy curve versus number of steps elapsed under Byzantine attacks, with a bigger client pool size.\label{fig:bztK=25}}
    \vspace{-10pt}
\end{figure}

\textbf{Settings.} We assume that there is $1$ Byzantine client and $4$ honest clients. The Byzantine client always transmits a random number as the gradient projection in \emph{ZO-FedSGD} and always transmits a reversed sign in \emph{FeedSign}. All other settings are consistent.

\textbf{Language models.} \Cref{table:byzantine_robustness} reports the test metric of \emph{ZO-FedSGD} and \emph{FeedSign} with one of the clients as a Byzantine client. The test metric of \emph{FeedSign} is higher than that of \emph{ZO-FedSGD} with the largest gap of $+6.5\%$. It establishes that \emph{FeedSign} expresses an inherent advantage in resisting Byzantine attacks.

\textbf{Image models.} We increase the client pool size to $K = 25$, and report the test loss and accuracy curves with $BK = 0$ to $3$ independent Byzantine clients fine-tuning a ViT-base model on the CIFAR-10 dataset, as in \Cref{fig:bztK=25}. It is observed that when there is no Byzantine attack (leftmost column), \emph{FeedSign} performs closely to \emph{ZO-FedSGD}. When the number of attackers increases, the performance of \emph{ZO-FedSGD} is gradually compromised. In contrast, \emph{FeedSign} maintains a robust performance, and its convergence is not compromised until there are $BK=3$ attackers.

\begin{remark}
    It should be noted that the simulated scenario in which the devices submit falsified gradient projections covers different attack methods, including gradient noise injection and label flipping. Specifically, for gradient noise injection, the gradient direction is determined by only the PRNG and the seed. Hence, the devices can only change the amplitude of the gradient. For label flipping, since the gradient estimate is projected to a determined direction, it is also equivalent to having an inaccurate gradient projection.
\end{remark}


\section{Discussions}

We include a brief discussion on some previously unrevealed features of federated learning systems based on \emph{seed-projection} or \emph{seed-sign pairs}, some of which may be shared by \emph{FeedSign} and \emph{ZO-FedSGD}. A detailed version can be found in the Appendix.

\textbf{Efficient Model Storage and Sharing.} The fine-tuned models can be represented with \emph{orbits}, namely, the elapsed \emph{seed-projection} or \emph{seed-sign} pairs from the starting checkpoint. For example, for a fine-tuned OPT-13B model, $24$GB of additional storage is required. However, the orbit generated by \emph{FeedSign} will occupy less than $200$ bytes of storage and guarantees perfect recovery of the fine-tuned model with $10 000$ fine-tune steps. The technique can largely reduce the load for model hubs like Huggingface.

\textbf{Parameter Servers can be Small and Task Agnostic.} Unlike conventional federated learning systems, both data and models are kept private, which circumvents the usually made but easily ignored assumption: the PS should be big and also a legal holder of the models.

\textbf{Memory-efficient Integration with Differential Privacy.} \emph{FeedSign} and \emph{ZO-FedSGD} can serve as an efficient framework that provides a strong privacy-convergence trade-off for different tasks with a slight modification on the aggregation rule since gradient statistics is not required.

\section{Conclusion}

We have presented a novel FFT framework \emph{FeedSign} that can operate in extremely deficient communication and memory budgets. Facilitated by ZO optimization and shared PRNG, each client needs only to upload one bit to the PS and then download one bit as a global update direction indicator in a step, and the memory overhead is equal to that needed for inference. We conduct theoretical analysis implying that \emph{FeedSign} has different kinds of robustness. Extensive experiments have shown that reducing communication overhead affects little on the performance of \emph{FeedSign} while bringing additional advantages.

\section*{Impact Statement}




This paper presents work whose goal is to reduce the impractically expensive costs of communication, memory, and computation in federated fine-tuning for large models. There are many potential societal consequences of our work, none of which we feel must be specifically highlighted here.

\nocite{langley00}

\bibliography{example_paper}

\begin{thebibliography}{89}
\providecommand{\natexlab}[1]{#1}
\providecommand{\url}[1]{\texttt{#1}}
\expandafter\ifx\csname urlstyle\endcsname\relax
  \providecommand{\doi}[1]{doi: #1}\else
  \providecommand{\doi}{doi: \begingroup \urlstyle{rm}\Url}\fi

\bibitem[Abadi et~al.(2016)Abadi, Agarwal, Barham, Brevdo, Chen, Citro,
  Corrado, Davis, Dean, Devin, et~al.]{abadi2016tensorflow}
Abadi, M., Agarwal, A., Barham, P., Brevdo, E., Chen, Z., Citro, C., Corrado,
  G.~S., Davis, A., Dean, J., Devin, M., et~al.
\newblock Tensorflow: Large-scale machine learning on heterogeneous distributed
  systems.
\newblock \emph{arXiv preprint arXiv:1603.04467}, 2016.

\bibitem[Achiam et~al.(2023)Achiam, Adler, Agarwal, Ahmad, Akkaya, Aleman,
  Almeida, Altenschmidt, Altman, Anadkat, et~al.]{achiam2023gpt}
Achiam, J., Adler, S., Agarwal, S., Ahmad, L., Akkaya, I., Aleman, F.~L.,
  Almeida, D., Altenschmidt, J., Altman, S., Anadkat, S., et~al.
\newblock Gpt-4 technical report.
\newblock \emph{arXiv preprint arXiv:2303.08774}, 2023.

\bibitem[Alistarh et~al.(2018)Alistarh, Allen-Zhu, and
  Li]{alistarh2018byzantine}
Alistarh, D., Allen-Zhu, Z., and Li, J.
\newblock Byzantine stochastic gradient descent.
\newblock \emph{Advances in neural information processing systems}, 31, 2018.

\bibitem[Allouah et~al.(2023)Allouah, Farhadkhani, Guerraoui, Gupta, Pinot, and
  Stephan]{allouah2023fixing}
Allouah, Y., Farhadkhani, S., Guerraoui, R., Gupta, N., Pinot, R., and Stephan,
  J.
\newblock Fixing by mixing: A recipe for optimal byzantine ml under
  heterogeneity.
\newblock In \emph{International Conference on Artificial Intelligence and
  Statistics}, pp.\  1232--1300. PMLR, 2023.

\bibitem[Amari(1993)]{amari1993backpropagation}
Amari, S.-i.
\newblock Backpropagation and stochastic gradient descent method.
\newblock \emph{Neurocomputing}, 5\penalty0 (4-5):\penalty0 185--196, 1993.

\bibitem[Bernstein et~al.(2018)Bernstein, Wang, Azizzadenesheli, and
  Anandkumar]{bernstein2018signsgd}
Bernstein, J., Wang, Y.-X., Azizzadenesheli, K., and Anandkumar, A.
\newblock signsgd: Compressed optimisation for non-convex problems.
\newblock In \emph{International Conference on Machine Learning}, pp.\
  560--569. PMLR, 2018.

\bibitem[Blanchard et~al.(2017)Blanchard, El~Mhamdi, Guerraoui, and
  Stainer]{blanchard2017machine}
Blanchard, P., El~Mhamdi, E.~M., Guerraoui, R., and Stainer, J.
\newblock Machine learning with adversaries: Byzantine tolerant gradient
  descent.
\newblock \emph{Advances in neural information processing systems}, 30, 2017.

\bibitem[Boopathy \& Fiete(2022)Boopathy and Fiete]{boopathy2022train}
Boopathy, A. and Fiete, I.
\newblock How to train your wide neural network without backprop: An
  input-weight alignment perspective.
\newblock In \emph{International Conference on Machine Learning}, pp.\
  2178--2205. PMLR, 2022.

\bibitem[Bottou(2010)]{bottou2010large}
Bottou, L.
\newblock Large-scale machine learning with stochastic gradient descent.
\newblock In \emph{Proceedings of COMPSTAT'2010: 19th International Conference
  on Computational StatisticsParis France, August 22-27, 2010 Keynote, Invited
  and Contributed Papers}, pp.\  177--186. Springer, 2010.

\bibitem[Bottou et~al.(2018)Bottou, Curtis, and
  Nocedal]{bottou2018optimization}
Bottou, L., Curtis, F.~E., and Nocedal, J.
\newblock Optimization methods for large-scale machine learning.
\newblock \emph{SIAM review}, 60\penalty0 (2):\penalty0 223--311, 2018.

\bibitem[Burkhart et~al.(2010)Burkhart, Strasser, Many, and
  Dimitropoulos]{burkhart2010sepia}
Burkhart, M., Strasser, M., Many, D., and Dimitropoulos, X.
\newblock $\{$SEPIA$\}$:$\{$Privacy-Preserving$\}$ aggregation of
  $\{$Multi-Domain$\}$ network events and statistics.
\newblock In \emph{19th USENIX Security Symposium (USENIX Security 10)}, 2010.

\bibitem[Chen et~al.(2023)Chen, Zhang, Jia, Diffenderfer, Liu, Parasyris,
  Zhang, Zhang, Kailkhura, and Liu]{chen2023deepzero}
Chen, A., Zhang, Y., Jia, J., Diffenderfer, J., Liu, J., Parasyris, K., Zhang,
  Y., Zhang, Z., Kailkhura, B., and Liu, S.
\newblock Deepzero: Scaling up zeroth-order optimization for deep model
  training.
\newblock \emph{arXiv preprint arXiv:2310.02025}, 2023.

\bibitem[Chen et~al.(2024{\natexlab{a}})Chen, Chen, Gu, and Deng]{chen2024fine}
Chen, J., Chen, H., Gu, B., and Deng, H.
\newblock Fine-grained theoretical analysis of federated zeroth-order
  optimization.
\newblock \emph{Advances in Neural Information Processing Systems}, 36,
  2024{\natexlab{a}}.

\bibitem[Chen et~al.(2024{\natexlab{b}})Chen, Huang, and Ye]{chen2024fair}
Chen, Y., Huang, W., and Ye, M.
\newblock Fair federated learning under domain skew with local consistency and
  domain diversity.
\newblock In \emph{Proceedings of the IEEE/CVF Conference on Computer Vision
  and Pattern Recognition}, pp.\  12077--12086, 2024{\natexlab{b}}.

\bibitem[Chiang et~al.(2022)Chiang, Ni, Miller, Bansal, Geiping, Goldblum, and
  Goldstein]{chiang2022loss}
Chiang, P.-y., Ni, R., Miller, D.~Y., Bansal, A., Geiping, J., Goldblum, M.,
  and Goldstein, T.
\newblock Loss landscapes are all you need: Neural network generalization can
  be explained without the implicit bias of gradient descent.
\newblock In \emph{The Eleventh International Conference on Learning
  Representations}, 2022.

\bibitem[Cho et~al.(2023)Cho, Liu, Xu, Fahrezi, Barnes, and
  Joshi]{cho2023heterogeneous}
Cho, Y.~J., Liu, L., Xu, Z., Fahrezi, A., Barnes, M., and Joshi, G.
\newblock Heterogeneous lora for federated fine-tuning of on-device foundation
  models.
\newblock In \emph{International Workshop on Federated Learning in the Age of
  Foundation Models in Conjunction with NeurIPS 2023}, 2023.

\bibitem[Cormode et~al.(2018)Cormode, Jha, Kulkarni, Li, Srivastava, and
  Wang]{cormode2018privacy}
Cormode, G., Jha, S., Kulkarni, T., Li, N., Srivastava, D., and Wang, T.
\newblock Privacy at scale: Local differential privacy in practice.
\newblock In \emph{Proceedings of the 2018 International Conference on
  Management of Data}, pp.\  1655--1658, 2018.

\bibitem[Delgado~Neto et~al.(2024)Delgado~Neto, Egger, Bakshi, and
  Bitar]{delgado2024communication}
Delgado~Neto, A. d.~S., Egger, M., Bakshi, M., and Bitar, R.
\newblock Communication-efficient byzantine-resilient federated zero-order
  optimization.
\newblock \emph{arXiv e-prints}, pp.\  arXiv--2406, 2024.

\bibitem[Dosovitskiy et~al.(2020)Dosovitskiy, Beyer, Kolesnikov, Weissenborn,
  Zhai, Unterthiner, Dehghani, Minderer, Heigold, Gelly,
  et~al.]{dosovitskiy2020image}
Dosovitskiy, A., Beyer, L., Kolesnikov, A., Weissenborn, D., Zhai, X.,
  Unterthiner, T., Dehghani, M., Minderer, M., Heigold, G., Gelly, S., et~al.
\newblock An image is worth 16x16 words: Transformers for image recognition at
  scale.
\newblock \emph{arXiv preprint arXiv:2010.11929}, 2020.

\bibitem[Fang et~al.(2022)Fang, Yu, Jiang, Shi, Jones, and
  Zhou]{fang2022communication}
Fang, W., Yu, Z., Jiang, Y., Shi, Y., Jones, C.~N., and Zhou, Y.
\newblock Communication-efficient stochastic zeroth-order optimization for
  federated learning.
\newblock \emph{IEEE Transactions on Signal Processing}, 70:\penalty0
  5058--5073, 2022.

\bibitem[Fang et~al.(2023)Fang, Ye, and Yang]{fang2023robust}
Fang, X., Ye, M., and Yang, X.
\newblock Robust heterogeneous federated learning under data corruption.
\newblock In \emph{Proceedings of the IEEE/CVF International Conference on
  Computer Vision}, pp.\  5020--5030, 2023.

\bibitem[Gardner(1984)]{gardner1984learning}
Gardner, W.~A.
\newblock Learning characteristics of stochastic-gradient-descent algorithms: A
  general study, analysis, and critique.
\newblock \emph{Signal processing}, 6\penalty0 (2):\penalty0 113--133, 1984.

\bibitem[Ghorbani et~al.(2019)Ghorbani, Krishnan, and
  Xiao]{ghorbani2019investigation}
Ghorbani, B., Krishnan, S., and Xiao, Y.
\newblock An investigation into neural net optimization via hessian eigenvalue
  density.
\newblock In \emph{International Conference on Machine Learning}, pp.\
  2232--2241. PMLR, 2019.

\bibitem[Goryczka \& Xiong(2015)Goryczka and Xiong]{goryczka2015comprehensive}
Goryczka, S. and Xiong, L.
\newblock A comprehensive comparison of multiparty secure additions with
  differential privacy.
\newblock \emph{IEEE transactions on dependable and secure computing},
  14\penalty0 (5):\penalty0 463--477, 2015.

\bibitem[Griewank(2014)]{griewank2014automatic}
Griewank, A.
\newblock On automatic differentiation and algorithmic linearization.
\newblock \emph{Pesquisa Operacional}, 34\penalty0 (3):\penalty0 621--645,
  2014.

\bibitem[Hard et~al.(2018)Hard, Rao, Mathews, Ramaswamy, Beaufays, Augenstein,
  Eichner, Kiddon, and Ramage]{hard2018federated}
Hard, A., Rao, K., Mathews, R., Ramaswamy, S., Beaufays, F., Augenstein, S.,
  Eichner, H., Kiddon, C., and Ramage, D.
\newblock Federated learning for mobile keyboard prediction.
\newblock \emph{arXiv preprint arXiv:1811.03604}, 2018.

\bibitem[Hu et~al.(2021)Hu, Shen, Wallis, Allen-Zhu, Li, Wang, Wang, and
  Chen]{hu2021lora}
Hu, E.~J., Shen, Y., Wallis, P., Allen-Zhu, Z., Li, Y., Wang, S., Wang, L., and
  Chen, W.
\newblock Lora: Low-rank adaptation of large language models.
\newblock \emph{arXiv preprint arXiv:2106.09685}, 2021.

\bibitem[Jamieson et~al.(2012)Jamieson, Nowak, and Recht]{jamieson2012query}
Jamieson, K.~G., Nowak, R., and Recht, B.
\newblock Query complexity of derivative-free optimization.
\newblock \emph{Advances in Neural Information Processing Systems}, 25, 2012.

\bibitem[Jiang et~al.(2023{\natexlab{a}})Jiang, Sablayrolles, Mensch, Bamford,
  Chaplot, Casas, Bressand, Lengyel, Lample, Saulnier,
  et~al.]{jiang2023mistral}
Jiang, A.~Q., Sablayrolles, A., Mensch, A., Bamford, C., Chaplot, D.~S., Casas,
  D. d.~l., Bressand, F., Lengyel, G., Lample, G., Saulnier, L., et~al.
\newblock Mistral 7b.
\newblock \emph{arXiv preprint arXiv:2310.06825}, 2023{\natexlab{a}}.

\bibitem[Jiang et~al.(2023{\natexlab{b}})Jiang, Roth, Li, Yang, Zhao, Nath, Xu,
  Dou, and Xu]{jiang2023fair}
Jiang, M., Roth, H.~R., Li, W., Yang, D., Zhao, C., Nath, V., Xu, D., Dou, Q.,
  and Xu, Z.
\newblock Fair federated medical image segmentation via client contribution
  estimation.
\newblock In \emph{Proceedings of the IEEE/CVF Conference on Computer Vision
  and Pattern Recognition}, pp.\  16302--16311, 2023{\natexlab{b}}.

\bibitem[Jin et~al.(2024)Jin, Liu, Huang, He, Wu, and Dai]{jin2024sign}
Jin, R., Liu, Y., Huang, Y., He, X., Wu, T., and Dai, H.
\newblock Sign-based gradient descent with heterogeneous data: Convergence and
  byzantine resilience.
\newblock \emph{IEEE Transactions on Neural Networks and Learning Systems},
  2024.

\bibitem[Kairouz et~al.(2021)Kairouz, McMahan, Avent, Bellet, Bennis, Bhagoji,
  Bonawitz, Charles, Cormode, Cummings, et~al.]{kairouz2021advances}
Kairouz, P., McMahan, H.~B., Avent, B., Bellet, A., Bennis, M., Bhagoji, A.~N.,
  Bonawitz, K., Charles, Z., Cormode, G., Cummings, R., et~al.
\newblock Advances and open problems in federated learning.
\newblock \emph{Foundations and trends{\textregistered} in machine learning},
  14\penalty0 (1--2):\penalty0 1--210, 2021.

\bibitem[Karimi et~al.(2016)Karimi, Nutini, and Schmidt]{karimi2016linear}
Karimi, H., Nutini, J., and Schmidt, M.
\newblock Linear convergence of gradient and proximal-gradient methods under
  the polyak-{\l}ojasiewicz condition.
\newblock In \emph{Machine Learning and Knowledge Discovery in Databases:
  European Conference, ECML PKDD 2016, Riva del Garda, Italy, September 19-23,
  2016, Proceedings, Part I 16}, pp.\  795--811. Springer, 2016.

\bibitem[Karimireddy et~al.(2020)Karimireddy, Kale, Mohri, Reddi, Stich, and
  Suresh]{karimireddy2020scaffold}
Karimireddy, S.~P., Kale, S., Mohri, M., Reddi, S., Stich, S., and Suresh,
  A.~T.
\newblock Scaffold: Stochastic controlled averaging for federated learning.
\newblock In \emph{International conference on machine learning}, pp.\
  5132--5143. PMLR, 2020.

\bibitem[Kim et~al.(2023)Kim, Kim, Mok, Park, and Lee]{kim2023client}
Kim, Y., Kim, J., Mok, W.-L., Park, J.-H., and Lee, S.
\newblock Client-customized adaptation for parameter-efficient federated
  learning.
\newblock In \emph{Findings of the Association for Computational Linguistics:
  ACL 2023}, pp.\  1159--1172, 2023.

\bibitem[Kingma(2014)]{kingma2014adam}
Kingma, D.~P.
\newblock Adam: A method for stochastic optimization.
\newblock \emph{arXiv preprint arXiv:1412.6980}, 2014.

\bibitem[Kone{\v{c}}n{\`y} et~al.(2016)Kone{\v{c}}n{\`y}, McMahan, Yu,
  Richt{\'a}rik, Suresh, and Bacon]{konevcny2016federated}
Kone{\v{c}}n{\`y}, J., McMahan, H.~B., Yu, F.~X., Richt{\'a}rik, P., Suresh,
  A.~T., and Bacon, D.
\newblock Federated learning: Strategies for improving communication
  efficiency.
\newblock \emph{arXiv preprint arXiv:1610.05492}, 2016.

\bibitem[Lang et~al.(2023{\natexlab{a}})Lang, Shlezinger, D'Oliveira, and
  Rouayheb]{lang2023compressed}
Lang, N., Shlezinger, N., D'Oliveira, R.~G., and Rouayheb, S.~E.
\newblock Compressed private aggregation for scalable and robust federated
  learning over massive networks.
\newblock \emph{arXiv preprint arXiv:2308.00540}, 2023{\natexlab{a}}.

\bibitem[Lang et~al.(2023{\natexlab{b}})Lang, Sofer, Shlezinger, D’Oliveira,
  and El~Rouayheb]{lang2023cpa}
Lang, N., Sofer, E., Shlezinger, N., D’Oliveira, R.~G., and El~Rouayheb, S.
\newblock Cpa: Compressed private aggregation for scalable federated learning
  over massive networks.
\newblock In \emph{ICASSP 2023-2023 IEEE International Conference on Acoustics,
  Speech and Signal Processing (ICASSP)}, pp.\  1--5. IEEE, 2023{\natexlab{b}}.

\bibitem[Li et~al.(2020)Li, Sahu, Zaheer, Sanjabi, Talwalkar, and
  Smith]{li2020federated}
Li, T., Sahu, A.~K., Zaheer, M., Sanjabi, M., Talwalkar, A., and Smith, V.
\newblock Federated optimization in heterogeneous networks.
\newblock \emph{Proceedings of Machine learning and systems}, 2:\penalty0
  429--450, 2020.

\bibitem[Li \& Liang(2021)Li and Liang]{li2021prefix}
Li, X.~L. and Liang, P.
\newblock Prefix-tuning: Optimizing continuous prompts for generation.
\newblock \emph{arXiv preprint arXiv:2101.00190}, 2021.

\bibitem[Ling et~al.(2024)Ling, Chen, Yao, Li, and Shen]{ling2024convergence}
Ling, Z., Chen, D., Yao, L., Li, Y., and Shen, Y.
\newblock On the convergence of zeroth-order federated tuning in large language
  models.
\newblock \emph{arXiv preprint arXiv:2402.05926}, 2024.

\bibitem[Liu et~al.(2019{\natexlab{a}})Liu, Chen, Chen, and
  Hong]{liu2019signsgd}
Liu, S., Chen, P.-Y., Chen, X., and Hong, M.
\newblock signsgd via zeroth-order oracle.
\newblock In \emph{International conference on learning representations},
  2019{\natexlab{a}}.

\bibitem[Liu et~al.(2019{\natexlab{b}})Liu, Ott, Goyal, Du, Joshi, Chen, Levy,
  Lewis, Zettlemoyer, and Stoyanov]{liu2019roberta}
Liu, Y., Ott, M., Goyal, N., Du, J., Joshi, M., Chen, D., Levy, O., Lewis, M.,
  Zettlemoyer, L., and Stoyanov, V.
\newblock Roberta: A robustly optimized bert pretraining approach.
\newblock \emph{arXiv preprint arXiv:1907.11692}, 2019{\natexlab{b}}.

\bibitem[Liu et~al.(2022)Liu, Chen, Ye, Fan, Li, and Li]{liu2022sash}
Liu, Z., Chen, S., Ye, J., Fan, J., Li, H., and Li, X.
\newblock Sash: Efficient secure aggregation based on shprg for federated
  learning.
\newblock In \emph{Uncertainty in Artificial Intelligence}, pp.\  1243--1252.
  PMLR, 2022.

\bibitem[Long et~al.(2020)Long, Tan, Jiang, and Zhang]{long2020federated}
Long, G., Tan, Y., Jiang, J., and Zhang, C.
\newblock Federated learning for open banking.
\newblock In \emph{Federated learning: privacy and incentive}, pp.\  240--254.
  Springer, 2020.

\bibitem[Malladi et~al.(2023)Malladi, Gao, Nichani, Damian, Lee, Chen, and
  Arora]{malladi2023fine}
Malladi, S., Gao, T., Nichani, E., Damian, A., Lee, J.~D., Chen, D., and Arora,
  S.
\newblock Fine-tuning language models with just forward passes.
\newblock \emph{Advances in Neural Information Processing Systems},
  36:\penalty0 53038--53075, 2023.

\bibitem[Mansouri et~al.(2023)Mansouri, {\"O}nen, Jaballah, and
  Conti]{mansouri2023sok}
Mansouri, M., {\"O}nen, M., Jaballah, W.~B., and Conti, M.
\newblock Sok: Secure aggregation based on cryptographic schemes for federated
  learning.
\newblock \emph{Proceedings on Privacy Enhancing Technologies}, 2023.

\bibitem[Maritan et~al.(2024)Maritan, Dey, and Schenato]{maritan2024fedzen}
Maritan, A., Dey, S., and Schenato, L.
\newblock Fedzen: Quadratic convergence in zeroth-order federated learning via
  incremental hessian estimation.
\newblock In \emph{2024 European Control Conference (ECC)}, pp.\  2320--2327.
  IEEE, 2024.

\bibitem[McMahan et~al.(2017)McMahan, Moore, Ramage, Hampson, and
  Arcas]{pmlr-v54-mcmahan17a}
McMahan, B., Moore, E., Ramage, D., Hampson, S., and Arcas, B. A.~y.
\newblock {Communication-Efficient Learning of Deep Networks from Decentralized
  Data}.
\newblock In Singh, A. and Zhu, J. (eds.), \emph{Proceedings of the 20th
  International Conference on Artificial Intelligence and Statistics},
  volume~54 of \emph{Proceedings of Machine Learning Research}, pp.\
  1273--1282. PMLR, 20--22 Apr 2017.
\newblock URL \url{https://proceedings.mlr.press/v54/mcmahan17a.html}.

\bibitem[Narayanan et~al.(2021)Narayanan, Shoeybi, Casper, LeGresley, Patwary,
  Korthikanti, Vainbrand, Kashinkunti, Bernauer, Catanzaro,
  et~al.]{narayanan2021efficient}
Narayanan, D., Shoeybi, M., Casper, J., LeGresley, P., Patwary, M.,
  Korthikanti, V., Vainbrand, D., Kashinkunti, P., Bernauer, J., Catanzaro, B.,
  et~al.
\newblock Efficient large-scale language model training on gpu clusters using
  megatron-lm.
\newblock In \emph{Proceedings of the International Conference for High
  Performance Computing, Networking, Storage and Analysis}, pp.\  1--15, 2021.

\bibitem[Nemirovskij \& Yudin(1983)Nemirovskij and
  Yudin]{nemirovskij1983problem}
Nemirovskij, A.~S. and Yudin, D.~B.
\newblock Problem complexity and method efficiency in optimization.
\newblock 1983.

\bibitem[Ning et~al.(2024)Ning, Wang, Qi, Zhu, Sun, Cheng, Liao, and
  Zhang]{ning2024fmdelta}
Ning, W., Wang, J., Qi, Q., Zhu, M., Sun, H., Cheng, D., Liao, J., and Zhang,
  C.
\newblock Fm-delta: Lossless compression for storing massive fine-tuned
  foundation models.
\newblock \emph{NeurIPS 2024}, 2024.

\bibitem[Ogier~du Terrail et~al.(2022)Ogier~du Terrail, Ayed, Cyffers,
  Grimberg, He, Loeb, Mangold, Marchand, Marfoq, Mushtaq,
  et~al.]{ogier2022flamby}
Ogier~du Terrail, J., Ayed, S.-S., Cyffers, E., Grimberg, F., He, C., Loeb, R.,
  Mangold, P., Marchand, T., Marfoq, O., Mushtaq, E., et~al.
\newblock Flamby: Datasets and benchmarks for cross-silo federated learning in
  realistic healthcare settings.
\newblock \emph{Advances in Neural Information Processing Systems},
  35:\penalty0 5315--5334, 2022.

\bibitem[Oktay et~al.(2020)Oktay, McGreivy, Aduol, Beatson, and
  Adams]{oktay2020randomized}
Oktay, D., McGreivy, N., Aduol, J., Beatson, A., and Adams, R.~P.
\newblock Randomized automatic differentiation.
\newblock \emph{arXiv preprint arXiv:2007.10412}, 2020.

\bibitem[Papyan(2020)]{papyan2020traces}
Papyan, V.
\newblock Traces of class/cross-class structure pervade deep learning spectra.
\newblock \emph{Journal of Machine Learning Research}, 21\penalty0
  (252):\penalty0 1--64, 2020.

\bibitem[Paszke et~al.(2017)Paszke, Gross, Chintala, Chanan, Yang, DeVito, Lin,
  Desmaison, Antiga, and Lerer]{paszke2017automatic}
Paszke, A., Gross, S., Chintala, S., Chanan, G., Yang, E., DeVito, Z., Lin, Z.,
  Desmaison, A., Antiga, L., and Lerer, A.
\newblock Automatic differentiation in pytorch.
\newblock 2017.

\bibitem[Paszke et~al.(2019)Paszke, Gross, Massa, Lerer, Bradbury, Chanan,
  Killeen, Lin, Gimelshein, Antiga, et~al.]{paszke2019pytorch}
Paszke, A., Gross, S., Massa, F., Lerer, A., Bradbury, J., Chanan, G., Killeen,
  T., Lin, Z., Gimelshein, N., Antiga, L., et~al.
\newblock Pytorch: An imperative style, high-performance deep learning library.
\newblock \emph{Advances in neural information processing systems}, 32, 2019.

\bibitem[Pfeiffer et~al.(2020)Pfeiffer, Kamath, R{\"u}ckl{\'e}, Cho, and
  Gurevych]{pfeiffer2020adapterfusion}
Pfeiffer, J., Kamath, A., R{\"u}ckl{\'e}, A., Cho, K., and Gurevych, I.
\newblock Adapterfusion: Non-destructive task composition for transfer
  learning.
\newblock \emph{arXiv preprint arXiv:2005.00247}, 2020.

\bibitem[Polyak(1964)]{polyak1964gradient}
Polyak, B.~T.
\newblock Gradient methods for solving equations and inequalities.
\newblock \emph{USSR Computational Mathematics and Mathematical Physics},
  4\penalty0 (6):\penalty0 17--32, 1964.

\bibitem[Popov et~al.(2018)Popov, Kudinov, Piontkovskaya, Vytovtov, and
  Nevidomsky]{popov2018distributed}
Popov, V., Kudinov, M., Piontkovskaya, I., Vytovtov, P., and Nevidomsky, A.
\newblock Distributed fine-tuning of language models on private data.
\newblock In \emph{International Conference on Learning Representations}, 2018.

\bibitem[Qin et~al.(2023)Qin, Chen, Qian, Ding, Li, and Deng]{qin2023federated}
Qin, Z., Chen, D., Qian, B., Ding, B., Li, Y., and Deng, S.
\newblock Federated full-parameter tuning of billion-sized language models with
  communication cost under 18 kilobytes.
\newblock \emph{arXiv preprint arXiv:2312.06353}, 2023.

\bibitem[Qiu et~al.(2023)Qiu, Shanbhag, and Yousefian]{qiu2023zeroth}
Qiu, Y., Shanbhag, U., and Yousefian, F.
\newblock Zeroth-order methods for nondifferentiable, nonconvex, and
  hierarchical federated optimization.
\newblock \emph{Advances in Neural Information Processing Systems}, 36, 2023.

\bibitem[Qu et~al.(2022)Qu, Zhou, Liang, Xia, Wang, Adeli, Fei-Fei, and
  Rubin]{qu2022rethinking}
Qu, L., Zhou, Y., Liang, P.~P., Xia, Y., Wang, F., Adeli, E., Fei-Fei, L., and
  Rubin, D.
\newblock Rethinking architecture design for tackling data heterogeneity in
  federated learning.
\newblock In \emph{Proceedings of the IEEE/CVF conference on computer vision
  and pattern recognition}, pp.\  10061--10071, 2022.

\bibitem[Ren et~al.(2022)Ren, Kornblith, Liao, and Hinton]{ren2022scaling}
Ren, M., Kornblith, S., Liao, R., and Hinton, G.
\newblock Scaling forward gradient with local losses.
\newblock \emph{arXiv preprint arXiv:2210.03310}, 2022.

\bibitem[Rieke et~al.(2020)Rieke, Hancox, Li, Milletari, Roth, Albarqouni,
  Bakas, Galtier, Landman, Maier-Hein, et~al.]{rieke2020future}
Rieke, N., Hancox, J., Li, W., Milletari, F., Roth, H.~R., Albarqouni, S.,
  Bakas, S., Galtier, M.~N., Landman, B.~A., Maier-Hein, K., et~al.
\newblock The future of digital health with federated learning.
\newblock \emph{NPJ digital medicine}, 3\penalty0 (1):\penalty0 1--7, 2020.

\bibitem[Rumelhart et~al.(1986)Rumelhart, Hinton, and
  Williams]{rumelhart1986learning}
Rumelhart, D.~E., Hinton, G.~E., and Williams, R.~J.
\newblock Learning representations by back-propagating errors.
\newblock \emph{nature}, 323\penalty0 (6088):\penalty0 533--536, 1986.

\bibitem[Sagun et~al.(2017)Sagun, Evci, Guney, Dauphin, and
  Bottou]{sagun2017empirical}
Sagun, L., Evci, U., Guney, V.~U., Dauphin, Y., and Bottou, L.
\newblock Empirical analysis of the hessian of over-parametrized neural
  networks.
\newblock \emph{arXiv preprint arXiv:1706.04454}, 2017.

\bibitem[So et~al.(2020)So, G{\"u}ler, and Avestimehr]{so2020byzantine}
So, J., G{\"u}ler, B., and Avestimehr, A.~S.
\newblock Byzantine-resilient secure federated learning.
\newblock \emph{IEEE Journal on Selected Areas in Communications}, 39\penalty0
  (7):\penalty0 2168--2181, 2020.

\bibitem[So et~al.(2021)So, G{\"u}ler, and Avestimehr]{so2021turbo}
So, J., G{\"u}ler, B., and Avestimehr, A.~S.
\newblock Turbo-aggregate: Breaking the quadratic aggregation barrier in secure
  federated learning.
\newblock \emph{IEEE Journal on Selected Areas in Information Theory},
  2\penalty0 (1):\penalty0 479--489, 2021.

\bibitem[Spall(1992)]{spall1992multivariate}
Spall, J.~C.
\newblock Multivariate stochastic approximation using a simultaneous
  perturbation gradient approximation.
\newblock \emph{IEEE transactions on automatic control}, 37\penalty0
  (3):\penalty0 332--341, 1992.

\bibitem[Sun et~al.(2024)Sun, Li, Li, and Ding]{sun2024improving}
Sun, Y., Li, Z., Li, Y., and Ding, B.
\newblock Improving lora in privacy-preserving federated learning.
\newblock \emph{arXiv preprint arXiv:2403.12313}, 2024.

\bibitem[Tang et~al.(2024)Tang, Panda, Nasr, Mahloujifar, and
  Mittal]{tang2024private}
Tang, X., Panda, A., Nasr, M., Mahloujifar, S., and Mittal, P.
\newblock Private fine-tuning of large language models with zeroth-order
  optimization.
\newblock \emph{arXiv preprint arXiv:2401.04343}, 2024.

\bibitem[Thapa et~al.(2022)Thapa, Arachchige, Camtepe, and
  Sun]{thapa2022splitfed}
Thapa, C., Arachchige, P. C.~M., Camtepe, S., and Sun, L.
\newblock Splitfed: When federated learning meets split learning.
\newblock In \emph{Proceedings of the AAAI Conference on Artificial
  Intelligence}, volume~36, pp.\  8485--8493, 2022.

\bibitem[Tian et~al.(2022)Tian, Cui, Liang, and Yu]{tian2022comprehensive}
Tian, Z., Cui, L., Liang, J., and Yu, S.
\newblock A comprehensive survey on poisoning attacks and countermeasures in
  machine learning.
\newblock \emph{ACM Computing Surveys}, 55\penalty0 (8):\penalty0 1--35, 2022.

\bibitem[Vahidian et~al.(2023)Vahidian, Morafah, Chen, Shah, and
  Lin]{vahidian2023rethinking}
Vahidian, S., Morafah, M., Chen, C., Shah, M., and Lin, B.
\newblock Rethinking data heterogeneity in federated learning: Introducing a
  new notion and standard benchmarks.
\newblock \emph{IEEE Transactions on Artificial Intelligence}, 5\penalty0
  (3):\penalty0 1386--1397, 2023.

\bibitem[Villalobos et~al.(2024)Villalobos, Ho, Sevilla, Besiroglu, Heim, and
  Hobbhahn]{villalobos2024will}
Villalobos, P., Ho, A., Sevilla, J., Besiroglu, T., Heim, L., and Hobbhahn, M.
\newblock Will we run out of data? limits of llm scaling based on
  human-generated data.
\newblock \emph{arXiv preprint arXiv:2211.04325}, pp.\  13--29, 2024.

\bibitem[Welling \& Teh(2011)Welling and Teh]{welling2011bayesian}
Welling, M. and Teh, Y.~W.
\newblock Bayesian learning via stochastic gradient langevin dynamics.
\newblock In \emph{Proceedings of the 28th international conference on machine
  learning (ICML-11)}, pp.\  681--688. Citeseer, 2011.

\bibitem[Wu et~al.(2020)Wu, Zhu, Wu, Wang, and Ge]{wu2020dissecting}
Wu, Y., Zhu, X., Wu, C., Wang, A., and Ge, R.
\newblock Dissecting hessian: Understanding common structure of hessian in
  neural networks.
\newblock \emph{arXiv preprint arXiv:2010.04261}, 2020.

\bibitem[Xu et~al.(2024)Xu, Cai, Wu, Li, and Wang]{xu2024fwdllm}
Xu, M., Cai, D., Wu, Y., Li, X., and Wang, S.
\newblock $\{$FwdLLM$\}$: Efficient federated finetuning of large language
  models with perturbed inferences.
\newblock In \emph{2024 USENIX Annual Technical Conference (USENIX ATC 24)},
  pp.\  579--596, 2024.

\bibitem[Yang et~al.(2018)Yang, Andrew, Eichner, Sun, Li, Kong, Ramage, and
  Beaufays]{yang2018applied}
Yang, T., Andrew, G., Eichner, H., Sun, H., Li, W., Kong, N., Ramage, D., and
  Beaufays, F.
\newblock Applied federated learning: Improving google keyboard query
  suggestions.
\newblock \emph{arXiv preprint arXiv:1812.02903}, 2018.

\bibitem[Yao et~al.(2020)Yao, Gholami, Keutzer, and Mahoney]{yao2020pyhessian}
Yao, Z., Gholami, A., Keutzer, K., and Mahoney, M.~W.
\newblock Pyhessian: Neural networks through the lens of the hessian.
\newblock In \emph{2020 IEEE international conference on big data (Big data)},
  pp.\  581--590. IEEE, 2020.

\bibitem[Ye et~al.(2023)Ye, Fang, Du, Yuen, and Tao]{ye2023heterogeneous}
Ye, M., Fang, X., Du, B., Yuen, P.~C., and Tao, D.
\newblock Heterogeneous federated learning: State-of-the-art and research
  challenges.
\newblock \emph{ACM Computing Surveys}, 56\penalty0 (3):\penalty0 1--44, 2023.

\bibitem[Yin et~al.(2018)Yin, Chen, Kannan, and Bartlett]{yin2018byzantine}
Yin, D., Chen, Y., Kannan, R., and Bartlett, P.
\newblock Byzantine-robust distributed learning: Towards optimal statistical
  rates.
\newblock In \emph{International conference on machine learning}, pp.\
  5650--5659. Pmlr, 2018.

\bibitem[Zaken et~al.(2021)Zaken, Ravfogel, and Goldberg]{zaken2021bitfit}
Zaken, E.~B., Ravfogel, S., and Goldberg, Y.
\newblock Bitfit: Simple parameter-efficient fine-tuning for transformer-based
  masked language-models.
\newblock \emph{arXiv preprint arXiv:2106.10199}, 2021.

\bibitem[Zelikman et~al.(2023)Zelikman, Huang, Liang, Haber, and
  Goodman]{zelikman2023just}
Zelikman, E., Huang, Q., Liang, P., Haber, N., and Goodman, N.~D.
\newblock Just one byte (per gradient): A note on low-bandwidth decentralized
  language model finetuning using shared randomness.
\newblock \emph{arXiv preprint arXiv:2306.10015}, 2023.

\bibitem[Zhang et~al.(2023{\natexlab{a}})Zhang, Li, Thekumparampil, Oh, and
  He]{zhang2023dpzero}
Zhang, L., Li, B., Thekumparampil, K.~K., Oh, S., and He, N.
\newblock Dpzero: Private fine-tuning of language models without
  backpropagation.
\newblock \emph{arXiv preprint arXiv:2310.09639}, 2023{\natexlab{a}}.

\bibitem[Zhang et~al.(2024)Zhang, Liu, Zhou, Du, Wei, Wang, and
  Chen]{zhang2024foresight}
Zhang, P., Liu, Y., Zhou, Y., Du, X., Wei, X., Wang, T., and Chen, M.
\newblock When foresight pruning meets zeroth-order optimization: Efficient
  federated learning for low-memory devices.
\newblock \emph{arXiv preprint arXiv:2405.04765}, 2024.

\bibitem[Zhang et~al.(2023{\natexlab{b}})Zhang, Yang, Dai, Wang, Yu, Qu, and
  Xu]{zhang2023fedpetuning}
Zhang, Z., Yang, Y., Dai, Y., Wang, Q., Yu, Y., Qu, L., and Xu, Z.
\newblock Fedpetuning: When federated learning meets the parameter-efficient
  tuning methods of pre-trained language models.
\newblock In \emph{Annual Meeting of the Association of Computational
  Linguistics 2023}, pp.\  9963--9977. Association for Computational
  Linguistics (ACL), 2023{\natexlab{b}}.

\end{thebibliography}
\bibliographystyle{icml2025}

\newpage
\appendix
\onecolumn

\section{Description of Symbols}

Descriptions of the symbols used in this paper can be found in \Cref{table:desc}.

\begin{table}[h]
    \setlength{\tabcolsep}{5pt}
    \caption{Descriptions of Symbols}
    \label{table:desc}
    \begin{center}
        \begin{tabular}{cl}
            \toprule
            Symbol & Description \\
            \hline
            $B$ & Batch size \\
            $\mathcal{B}$ & Data batch \\
            $c_g, \sigma_g$ & Batch gradient estimation noise factor \\
            $c_h, \sigma_h$ & Client-wise gradient estimation noise factor \\
            $\mathcal{D}$ & Dataset \\
            $d$ & Number of model parameters \\
            $k$ & Index of clients \\
            $K$ & Number of clients in an FL system \\
            $\mathcal{L}$ & Loss function \\
            $\mathcal{L}^*$ & Infimum value of loss function \\
            $L$ & Smooth constant of the loss function \\
            $n$ & Number of SPSA samples \\
            $\mathcal{N}$ & Gaussian distribution \\
            $p_{t,b}$ & Probability of a client being a Byzantine client at round $t$ \\
            $p_{t,e}$ & Inherent probability of a batch gradient estimate having a reversed sign at round $t$ \\
            $p_t$ & Overall probability of a batch gradient estimate having a reversed sign \\
            $s$ & Random seed \\
            $T$ & Number of global steps (step budget) \\
            $t$ & Index of global epochs (the number of total communication rounds) \\
            $\boldsymbol{w}$ & Model parameter vector \\
            $\alpha$ & Controlling parameter between gradient variance and optimization gap \\
            $\delta$ & Polyak-\L ojaciewicz property constant \\
            $\epsilon$ & Toleration threshold of the gap to error floor \\
            $\mu$ & Perturbation scale for forward passes-only fine-tuning \\
            $\zeta$ & Low-effective rank factor of the gradient estimator \\
            $\eta$ & Learning rate \\
            $\nabla$ & Gradient operator \\
            $\mathbb{E}$ & Expectation operator \\
            $\mathbb{V}$ & Variance operator \\
            $\mathbb{R}^n$ & $n$-dimensional real number set \\
            $\langle \cdot, \cdot \rangle$ & Inner product \\
            $tr$ & Trace operator \\
            $\|\cdot\|_{\text{op}}$ & Operator norm of matrices \\
            \hline
            $\mathcal{L}(\boldsymbol{w})$ & Loss function at model parameter $\boldsymbol{w}$ \\
            $\mathcal{L}_k(\boldsymbol{w})$ & Loss function of client $k$ at model parameter $\boldsymbol{w}$ \\
            $\hat{\mathcal{L}}_k(\boldsymbol{w}, \mathcal{B})$ & Loss function measured on data batch $\mathcal{B}$ at model parameter $\boldsymbol{w}$ on client $k$ \\
            $\mathcal{N}(\boldsymbol{\mu}, \boldsymbol{\Sigma})$ & Multivariate Gaussian distribution with center $\boldsymbol{\mu}$ and $\boldsymbol{\Sigma}$ \\
            \bottomrule
        \end{tabular}
    \end{center}
\end{table}

\section{Proof Sketch}

\subsection{Descent Lemma of \emph{FeedSign}}

Note that since \emph{FeedSign} is not an unbiased gradient estimator, one can not start from Lemma \ref{lm:descent} to reach this theorem. So we will need a different descent lemma for the global convergence bound.

\begin{theorem}[Dimension-free Descent Bound for \emph{FeedSign}]\label{thm:feedsign_dsc}
Given $\mathcal{L}(\boldsymbol{w})$ is a $L$-smooth function, the expected per-step loss descent can be bounded as follows: \begin{equation}\begin{aligned}\label{eq:lm3}
    \mathbb{E}(\mathcal{L}(\boldsymbol{w}_{t+1})) \leq \mathcal{L}(\boldsymbol{w}_t) - \eta (1 - 2p_t) \sqrt{\frac{2}{\pi}} \|\nabla\mathcal{L}(\boldsymbol{w}_t)\|_2 + \frac{\eta^2 L r}{2}.
\end{aligned}\end{equation}
where $\pi$ is the circumference ratio.
\end{theorem}

\begin{proof}
    Start from \Cref{ass:l-smooth}. Then, two key steps to arrive at our target are to calculate the inner product term and the quadratic term. 
    
    For the inner product term, the weight distance $\boldsymbol{w}_{t+1} - \boldsymbol{w}_t$ term will be replaced with an estimated gradient projection-related term. Note that $\text{Sign}(p) = p / |p|$, we can replace the weight distance with a term containing the unbiased gradient estimator $\nabla\mathcal{L}(\boldsymbol{w}_t, \mathcal{B})$. \Cref{ass:srp} further replaces it with a true gradient-related term. This term is eventually half-normal distributed, whose expectation can be easily derived. 
    
    For the quadratic norm term, since \emph{FeedSign} records no amplitude, the norm of the weight difference $\|\boldsymbol{w}_{t+1} - \boldsymbol{w}_t\|_2^2$ is fixed to a constant related to the Lipschitz constant $L$ and learning rate $\eta$. We will reach \Cref{thm:feedsign_dsc} after this.
\end{proof}

\begin{wrapfigure}{r}{0.4\textwidth}
    \vspace{-20pt}
    \centering
    \includegraphics[width=0.35\textwidth]{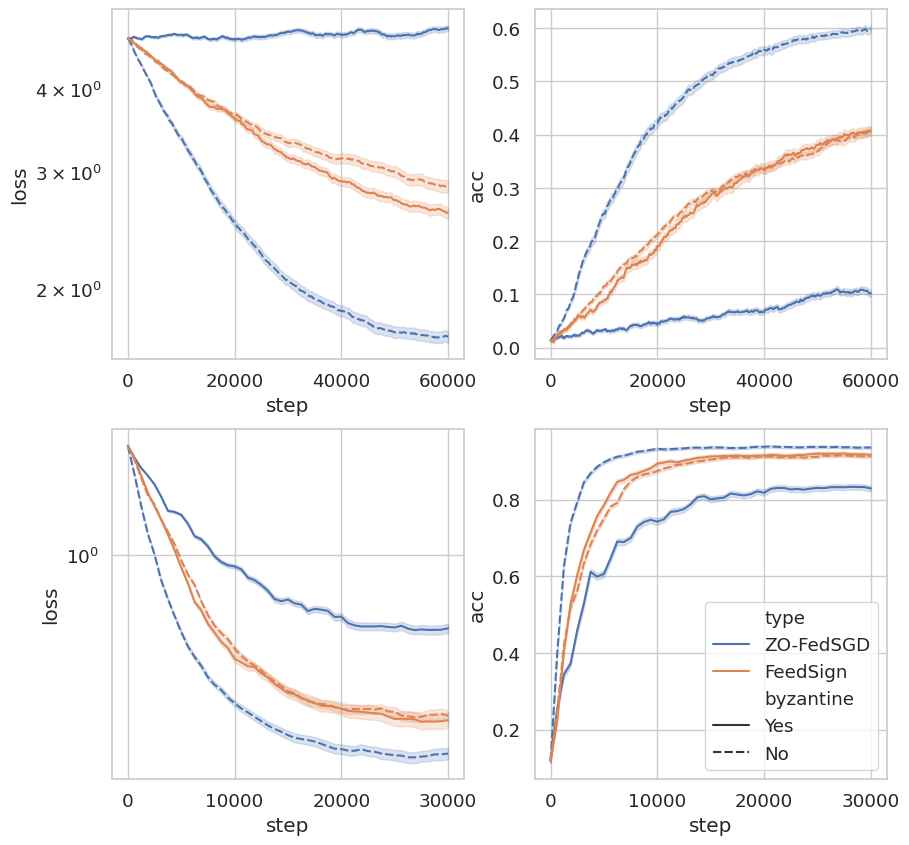}
    \vspace{-10pt}
    \caption{Loss and accuracy curve versus number of steps elapsed under Byzantine attacks.\label{fig:bzt}}
    \vspace{-50pt}
\end{wrapfigure}

With the stepwise loss descent bound of both \emph{ZO-FedSGD} and \emph{FeedSign}, we derive the global convergence in terms of expected loss.

\subsection{Proof Sketch to \Cref{thm:1}}

\begin{proof}
    Subtract $\mathcal{L}^*$ on both sides of the stepwise loss descent bound, and replace the quadratic gradient term to a term containing $\mathcal{L}(\boldsymbol{w}_t) - \mathcal{L}^*$. Then balance the factors so that a contraction inequality on the distance from $\mathcal{L}(\boldsymbol{w}_t)$ to its lower bound is reached. Iterate the contraction to harvest an exponential decay of the expected loss.
\end{proof}

\section{More Results}

\subsection{General Settings}


\textbf{Language Models.} As is done in \emph{MeZO}, we run few-shot learning for classification tasks on RoBERTa-large with $k=16$ samples per category, and general fine-tuning on OPT models. We employ test accuracy as the metric for classification and multiple-choice tasks and F1 score for generation tasks. Results are reported in \Cref{table:nlp_roberta_results}. We also report the results of the same task packages on OPT-125M in \Cref{table:differentK}.

\begin{table*}[t]
    \setlength{\tabcolsep}{5pt}
    \caption{Results on RoBERTa-large over language tasks. The best results obtained using federated ZO optimization is \textbf{bolded}, and the metric gap to that of the FO method is reported in the rightmost column.}
    \label{table:nlp_roberta_results}
    \begin{center}
        \begin{tabular}{lccccccc}
            \toprule
            \multicolumn{1}{c}{Task}  &\multicolumn{1}{c}{\bf SST-2}  &\multicolumn{1}{c}{\bf SST-5}  &\multicolumn{1}{c}{\bf SNLI}  &\multicolumn{1}{c}{\bf MNLI} &\multicolumn{1}{c}{\bf RTE}  &\multicolumn{1}{c}{\bf TREC} & \multirow{2}*{Gap} \\ 
            \multicolumn{1}{c}{Type} & \multicolumn{2}{c}{---- sentiment ----} & \multicolumn{3}{c}{- natural language inference -} & \multicolumn{1}{c}{-- topic --} & \\
            \hline
            Zero-shot & 79.0 & 35.5 & 50.2 & 48.8 & 51.4 & 32.0 & -- \\
            \hline
            \multicolumn{7}{c}{$\text{number of shots}=16$} \\
            \hdashline
            FO & 91.8 & 47.5 & 77.5 & 70.0 & 66.4 & 85.0 & -- \\
            \hdashline
            MeZO & 90.5 & 45.5 & 68.5 & 58.7 & 64.0 & 76.9 & -5.6 \\
            \hdashline
            \emph{ZO-FedSGD} & 88.9	& \bf 44.2&	67.9&	59.0&	60.3&	74.2&	-7.2 \\
            &(1.5)&	(2.0)&	(4.0)&	(5.3)&	(5.9)&	(3.6)& -- \\
            \bf \emph{FeedSign}& \bf 89.9&	\bf 44.2&	\bf 70.6&	\bf 64.2&	\bf 63.7&	\bf 78.4&	\bf -4.4  \\
            &(1.5)&	(0.7)&	(1.0)&	(2.6)&	(2.7)&	(3.0)& -- \\
            \bottomrule
        \end{tabular}
    \end{center}
\end{table*}

\textbf{Image models.} However, we also notice that for the last-layer FFT on a ViT-large model, although \emph{FeedSign} performs closely to \emph{ZO-FedSGD}, it cannot outperform. We infer that this could be accounted for by the good feature extraction ability of ViT models.

\begin{table*}[t]
    \setlength{\tabcolsep}{2pt}
    \caption{Main results on OPT-125M over language models with iid with different sizes of client pool.}
    \label{table:differentK}
    \begin{center}
        \begin{tabular}{lrccccccccccc}
            \toprule
            \multicolumn{1}{c}{Task} &\multicolumn{1}{c}{$K$} &\multicolumn{1}{c}{\bf SST-2}  &\multicolumn{1}{c}{\bf RTE}  &\multicolumn{1}{c}{\bf CB}  &\multicolumn{1}{c}{\bf BoolQ} &\multicolumn{1}{c}{\bf WSC}  &\multicolumn{1}{c}{\bf WIC} &\multicolumn{1}{c}{\bf MultiRC} &\multicolumn{1}{c}{\bf COPA}  &\multicolumn{1}{c}{\bf ReCoRD} &\multicolumn{1}{c}{\bf SQuAD}  &\multicolumn{1}{c}{\bf DROP} \\
            \hline
            Zero-shot & - & 51.2 & 53.0 & 48.2 & 41.5 & 37.5 & 51.2 & 49.7 & 69.0 & 51.7 & 9.5 & 4.4 \\
            \hline
            MeZO & - & 82.2 & 55.9 & 67.8 & 61.0 & 59.6 & 51.0 & 53.3 & 68.0 & 47.1 & 44.1 & 15.2 \\
            \hline

            \multirow{4}{*}{\emph{ZO-FedSGD}} & \multirow{2}{*}{5} & 84.2&	55.7&	66.7&	59.8&	\bf 56.0&	52.0&	55.1&	\bf 63.2&	\bf 48.4&	46.1&	17.2 \\
            & & (0.6)&	(1.9)&	(2.4)&	(1.4)&	(4.8)&	(1.5)&	(2.7)&	(3.0)&	(2.1)&	(1.2)&	(1.7) \\
             & \multirow{2}{*}{25} & 81.5&	52.1&	66.7&	59.4&	51.1&	52.2&	53.1&	\bf 66.2&	\bf 50.4&	37.5&	13.7  \\
             & & (1.8)&	(1.6)&	(2.0)&	(1.9)&	(9.8)&	(0.5)&	(4.5)&	(1.6)&	(2.0)&	(2.9)&	(1.4) \\
            \hdashline
            \multirow{4}{*}{\textit{\textbf{FeedSign}}} & \multirow{2}{*}{5} & \bf 84.7&	\bf 58.8&	\bf 67.0&	\bf 60.2&	52.4&	\bf 56.2&	\bf 60.7&	60.6&	46.4&	\bf 47.4&	\bf 18.2 \\
             &  & (0.7)&	(3.3)&	(2.0)&	(1.9)&	(8.2)&	(1.2)&	(1.7)&	(2.0)&	(2.1)&	(0.5)&	(0.3) \\
              & \multirow{2}{*}{25} & \bf 84.7&	\bf 56.6&	\bf 67.4&	\bf 60.8&	\bf 54.9&	\bf 54.0&	\bf 56.5&	64.8&	\bf 50.4&	\bf 45.6&	\bf 16.1 \\
             &  & (0.5)&	(3.2)&	(0.8)&	(1.1)&	(6.7)&	(1.5)&	(2.4)&	(1.6)&	(1.5)&	(0.8)&	(1.0) \\
            \bottomrule
    \end{tabular}
    \end{center}
\end{table*}

\subsection{Byzantine Resilience}

\vspace{-60pt}
\begin{wraptable}{r}{0.4\linewidth}
\centering
\caption{Results on ViT-large FFT, client pool size $K=5$, with one Byzantine attacker.}
\begin{tabular}{@{}lcc@{}}
\toprule
 & \bf CIFAR-10 & \bf CIFAR-100 \\ \midrule
\emph{ZO-FedSGD} &  83.9 (9.8)     & 10.9 (2.5)     \\
\emph{FeedSign}  & \bf 91.9 (6.7)     & \bf 40.8 (5.0)     \\ \bottomrule
\end{tabular}\label{table:vision models byzantine}
\vspace{-30pt}
\end{wraptable}

\Cref{table:vision models byzantine} and \Cref{fig:bzt} report the test accuracy with $1$ of the $5$ clients as a Byzantine client fine-tuning a ViT-large model. It can be observed that \emph{ZO-FedSGD} is completely compromised with the Byzantine attack, while \emph{FeedSign} maintains its performance.

\begin{wrapfigure}{r}{0.4\textwidth}
    \vspace{-40pt}
    \includegraphics[width=1\linewidth]{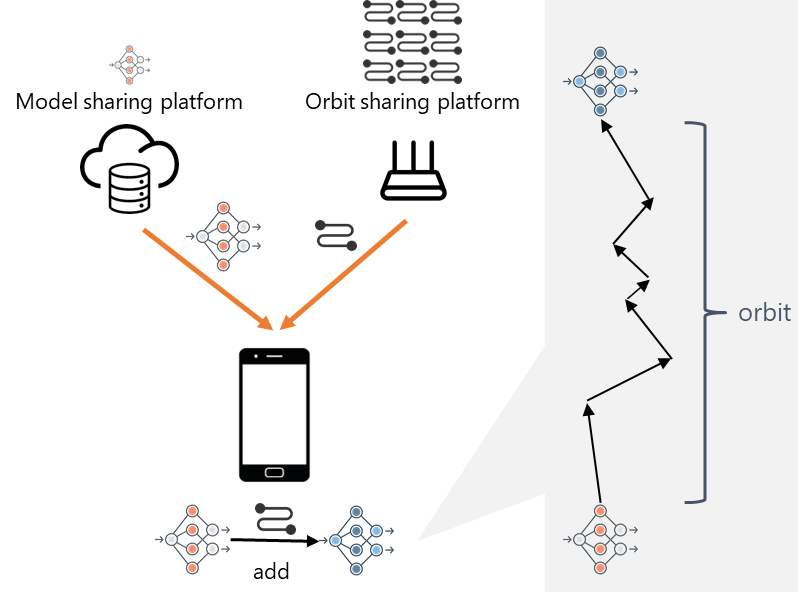}
    \vspace{-10pt}
    \caption{Orbits-based efficient model storage and sharing.\label{fig:orbits}}
    \vspace{-10pt}
\end{wrapfigure}

\section{Expanded Discussions}

With the performance of \emph{FeedSign} well evaluated, we look further for some byproducts brought by the design of the framework.

\subsection{Efficient Model Storage and Sharing}

\label{sec:faster}

It is estimated that over $600, 000$ models are stored in model sharing platforms like Huggingface, $90 \%$ of them are fine-tuned models \cite{ning2024fmdelta}. Frequently moving them results in PBs of monthly information transmission and storage demand. Notably, the platform can save only a small number of well-recognized checkpoints and save the \emph{orbits}, which is the collection of seed-projection pairs elapsed from a checkpoint to fine-tuned models by using seed-projection pairs as in \emph{ZO-FedSGD} or seed-sign pairs as in \emph{FeedSign}, as shown in \Cref{fig:orbits}. For example, for a fine-tuned OPT-13B model with $10, 000$ fine-tune steps, $24$GB of additional storage is required. However, the orbit generated by \emph{FeedSign} will occupy less than $200$ bytes of storage and guarantees perfect recovery of the fine-tuned model.



\subsection{Parameter Servers can be Small and Task Agnostic}

    
\begin{wrapfigure}{r}{0.4\textwidth}
    \includegraphics[width=1\linewidth]{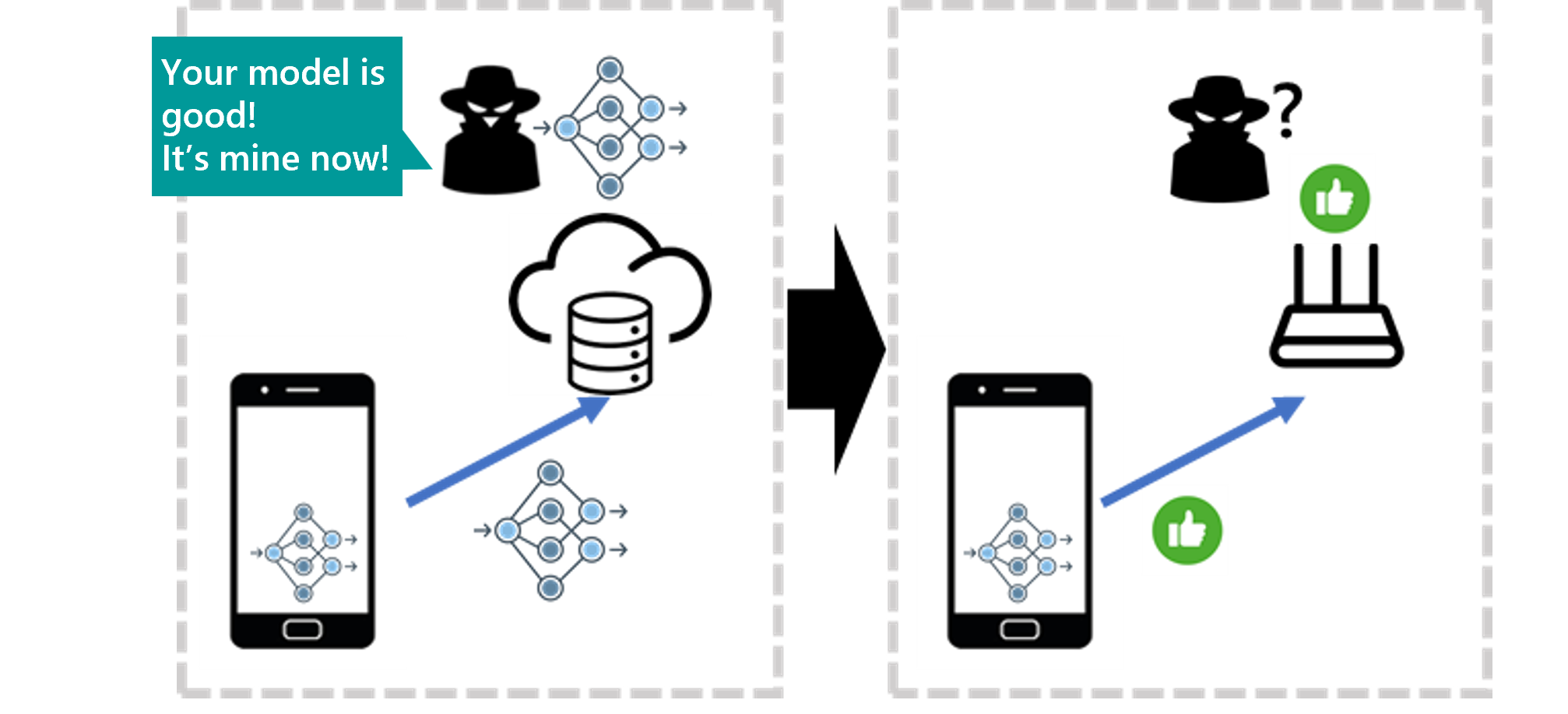}
    \caption{Orbits-based efficient model storage and sharing.\label{fig:comme}}
\end{wrapfigure}


A byproduct of PS holding no actual DL model parameter of \emph{FeedSign} is parameter security. 
This is because if operating \textit{FedAvg} without special design like homomorphic encryption \cite{liu2022sash}, \cite{mansouri2023sok}, generic secure multiparty computation \cite{burkhart2010sepia}, or additive masks \cite{so2021turbo, goryczka2015comprehensive}, the PS always knows the model parameters and hence has to be a legal holder of the final model. 
This is because if the PS aggregates model parameters, the final trained model is trivial for the PS to take. If the PS aggregates model updates, as long as the PS knows the initial model parameter, which is often assumed so, it can reconstruct the final trained model by recording every update made by the clients. Since the PS is usually assumed to be a data center or cloud PS in conventional FL system designs, the PS has sufficient resources to do this. However, \textbf{not only data but also models are kept private and local} in FL systems featuring designs similar to \emph{FeedSign}, as shown in \Cref{fig:comme}.

In fact, according to Section \ref{sec:faster}, the PS can be a device that is too small to host the actual model. Moreover, conventional model-sharing platforms need to maintain large storage to store millions of models. However, with a \emph{FeedSign}-like seed-projection pairs design, the platform will not need to store the actual parameters but only the orbits of elapsed seed-projection pairs during fine-tuning from some well-recognized checkpoints.

Additionally, in conventional designs, fortuitous clients that join the FL system midway should download the newest global model from the PS; this may lead to a global slowdown. Otherwise, the newly joining client could have downloaded an outdated model. In \emph{FeedSign}, the PS records and sends the \textit{orbit}, which is the collection of seed and projected gradient pairs generated by the clients since the start of the fine-tuning to the fortuitous clients, since there are no actual model parameters on the PS.

\subsection{Light-weight Integration with Differential Privacy}

\emph{FeedSign} can serve as an extremely memory-efficient framework that provides a strong privacy-convergence trade-off for different task requirements with a small modification on the aggregation rule.

\begin{definition}[Differentially Private Update Aggregation]\label{def:dp_agg}
    The global model of FL updates with learning rate $\eta$ under the following rule: \begin{align}
        (\textbf{\textit{DP-FeedSign}}) \quad \boldsymbol{w} \gets \boldsymbol{w} -  f_{\text{DP}}(p_1, \dots, p_K) \eta\boldsymbol{z}.
    \end{align}
    where $f_{\text{DP}}$ is a random variable with probability \begin{align}
        \text{Prob}(f_{\text{DP}} = 1) &= p_+ / (p_+ + p_-), \\
        \text{Prob}(f_{\text{DP}} = -1) &= p_- / (p_+ + p_-),
    \end{align} where \begin{align}
        &p_{\pm} = \exp\left(\frac{\epsilon q_\pm}{4}\right), \quad q_\pm = \sum_{k=1}^K \left(\frac12 \pm \frac{p_k}{|p_k|}\right)
    \end{align} with $K$ participating clients.
\end{definition}

\begin{theorem}[Differential Privacy Guarantee]\label{thm:dp}
    Algorithm \ref{alg:feedavg} with its update rule replaced as Definition \ref{def:dp_agg} is $(\epsilon, 0)$-DP.
\end{theorem}

\begin{remark}
    By pushing $\epsilon$ to $0$, we will have a stronger differential privacy (DP) guarantee, while the behavior of $f_{\text{DP}}$ will become more similar to $\text{Bernoulli}(0.5)$. This will result in $p_t$ in \Cref{thm:1} approaching $1 / 2$, slowing down the convergence of \emph{FeedSign}.
\end{remark}

\begin{remark}
    Like \cite{tang2024private}, our DP follows a new mechanism by only privatizing the gradient projection while it differs by having a discrete output. This is based on the fact that with the seed being broadcast and all machines sharing the same PRNG, the only uncertainty about the gradient for a malicious user is the sign of the corresponding gradient projection.
\end{remark}

\subsection{Impact of Byzantine Attacks on \emph{FeedSign}}

A characterization of how the ratio of Byzantine clients impacts \emph{FeedSign} is as follows.

\begin{proposition}[Reversed Sign Probability with Byzantine Clients] The batch gradient estimator $\hat{\nabla}\mathcal{L}_k(\boldsymbol{w}_t, \mathcal{B})$ will have a reversed sign to the true gradient $\nabla \mathcal{L}$ with a probability of \begin{align}
    p_t = p_{t, e} + p_{t, b} - p_{t, e}p_{t, b},
\end{align}
where $p_{t, e}$ is the inherent reversed sign probability due to batch gradient estimation error and $p_{t, b}$ is the proportion of Byzantine clients at step $t$.
\label{thm:reverse_prob}
\end{proposition}

\section{The Range of Sign-Reversing Probability}

As stated in \Cref{thm:1}, the convergence of \emph{FeedSign} is dependent on the sign-reversing probability $p_t$, so knowledge of the range of this value is necessary for us to understand the convergence behavior of \emph{FeedSign}. Since the impact of Byzantine attacks is captured in \Cref{thm:reverse_prob}, we will focus on $p_{t,e}$.

\begin{wrapfigure}{r}{0.3\textwidth}
    \includegraphics[width=1\linewidth]{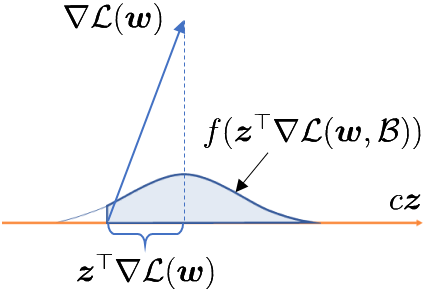}
    \vspace{-10pt}
    \caption{Inherent sign-reversing probability density.\label{fig:pte}}
    \vspace{-20pt}
\end{wrapfigure}

\subsection{Analysis}

Following most of the research works, we assume the batch sampling is unbiased, so that $\mathbb{E}_{\mathcal{B}}[\nabla \mathcal{L}(\boldsymbol{w}, \mathcal{B})] = \nabla\mathcal{L}(\boldsymbol{w})$. To utilize an analysis on $p_{t, e}$, we make the following assumption.

\begin{assumption}\label{ass:pte}
    Given model $\boldsymbol{w}$, loss function $\mathcal{L}$, and a random direction $\boldsymbol{z}$ satisfying all necessary conditions. The distribution of gradient projection estimation $\boldsymbol{z}^\top\mathcal{L}(\boldsymbol{w}, \mathcal{B})$ with batch $\mathcal{B}$ is symmetric on the randomness of batch sampling.
\end{assumption}

\begin{proposition}\label{prop:pte}
    The inherent sign-reversing probability is \begin{align}
        p_{t, e} = \begin{cases}
            F(0), & \text{if }\theta \leq \frac{\pi}{2} \\
            1-F(0), & \text{otherwise}
        \end{cases},
    \end{align}
    so $p_{t,e} \leq 1/2$ and equality holds if and only if $\theta = \frac{\pi}{2}$, with $F$ being the cumulative density function of $\boldsymbol{z}^\top \nabla \mathcal{L}(\boldsymbol{w}, \mathcal{B})$ and $\theta$ begin the angle between $\boldsymbol{z}$ and $\nabla \mathcal{L}(\boldsymbol{w})$.
\end{proposition}

\begin{proof}
    Consider the 2-dimensional plane spanned by $\boldsymbol{z}$ and $\nabla \mathcal{L}(\boldsymbol{w})$. Noticing that the desired event corresponds to the blank part of the area under the density function, we arrive at the result stated above.
\end{proof}

This property will guarantee that $0 < A < 1$ in \Cref{thm:1} is always true if there are no Byzantine attackers, and contributes to the robustness of \emph{FeedSign}.

\subsection{Simulations}

We run a small simulation to support \Cref{prop:pte}. We run OPT-125m on the SST2 task. We choose a training set consisting of $5000$ samples, and for every $4000$ steps, we evaluate the gradient projection by sample for the gradient direction corresponding to seeds $s = 0$ to $39$. We average the gradient projections to obtain $\boldsymbol{z}_s^\top \mathcal{L}(\boldsymbol{w}_t)$. We then uniformly sample $10000$ batches of size $64$ and take the average of the gradient projection $\boldsymbol{z}_s^\top \mathcal{L}(\boldsymbol{w}_t, \mathcal{B})$ by sample within each batch $\mathcal{B}$. We compute $p_{t,e}$ as the proportion of the batches holding a gradient projection by batch against $\boldsymbol{z}_s^\top\nabla \mathcal{L}(\boldsymbol{w}_t)$ given step $t$ and seed $s$.

\subsubsection{Range of Inherent Sign-Reversing Probability}

We report the measured $p_{t, e}$ as \Cref{fig:pte_sim1}.

\begin{figure}[h]
    \centering
    \includegraphics[width=1\linewidth]{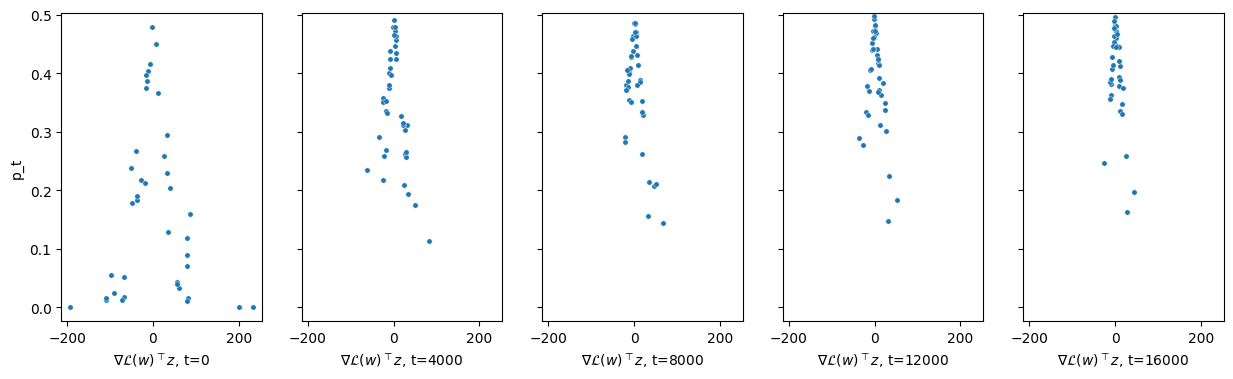}
    \caption{Inherent sign-reversing probability vs. $\boldsymbol{z}_s^\top\nabla\mathcal{L}(\boldsymbol{w})$.}
    \label{fig:pte_sim1}
\end{figure}

We noticed that the gradient projections are generally small, which is due to the fact that the space dimension is high. The measured highest reading of $p_{t,e} = 0.4968$ is at $t = 12000$, with $\boldsymbol{z}^\top \nabla \mathcal{L}(\boldsymbol{w}) = -0.3330$.

\subsubsection{Symmetric Distribution of $\boldsymbol{z}^\top \mathcal{L}(\boldsymbol{w}, \mathcal{B})$}

We report $\boldsymbol{z}_s^\top \nabla\mathcal{L}(\boldsymbol{w}, \mathcal{B})$ from $s=0$ to $4$, $t=0, 4000, 8000, 12000, 16000$ in \Cref{fig:pte_sim2}. The red lines marks the corresponding $\boldsymbol{z}^\top \nabla \mathcal{L}(\boldsymbol{w})$. The distributions manifest an obvious symmetric pattern against the randomness of batch sampling as an empirical support to \cref{ass:pte}.

\begin{figure}[h]
    \centering
    \includegraphics[width=1\linewidth]{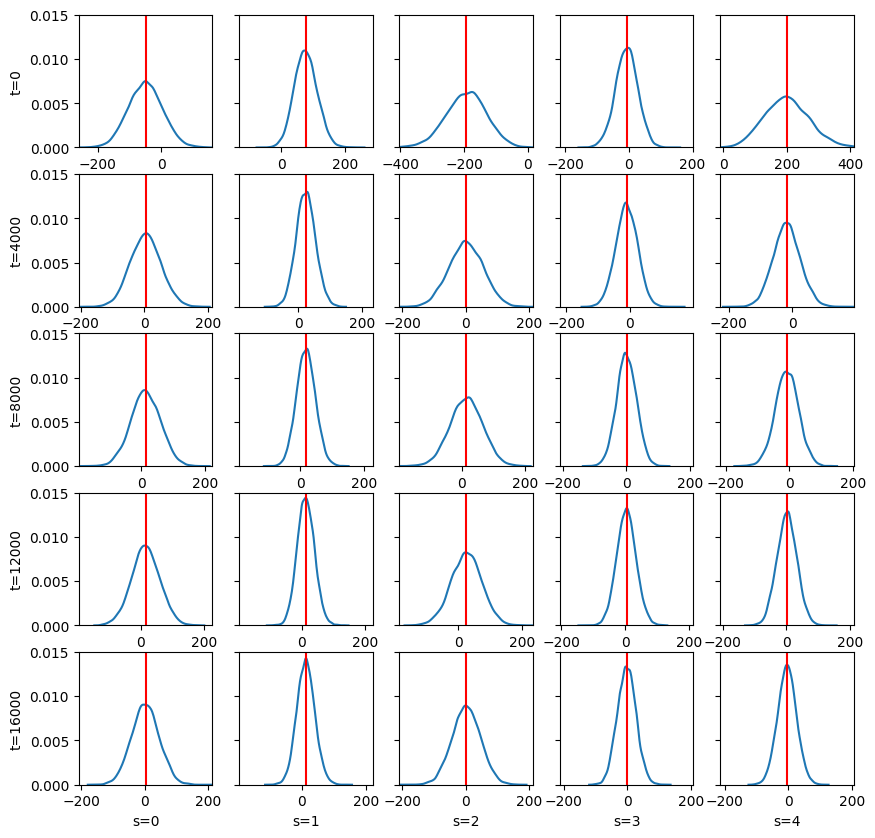}
    \caption{Distribution of $\boldsymbol{z}_s^\top\nabla\mathcal{L}(\boldsymbol{w}, \mathcal{B})$. The red lines are their corresponding $\boldsymbol{z}^\top \nabla\mathcal{L}(\boldsymbol{w})$.}
    \label{fig:pte_sim2}
\end{figure}

\section{Proofs}

\label{appd:proofs}

\subsection{Proof to \Cref{thm:feedsign_dsc}}

\begin{proof}

Following the first lines of Proof of Theorem 1 from \cite{malladi2023fine} with \Cref{ass:r-effective}, we have \begin{align}
    \mathcal{L}(\boldsymbol{w}_{t+1}) \leq \mathcal{L}(\boldsymbol{w}_t) - \eta \nabla \mathcal{L}(\boldsymbol{w}_{t})^\top (\boldsymbol{w}_{t+1} - \boldsymbol{w}_t) + \frac12 \eta^2 (\boldsymbol{w}_{t+1} - \boldsymbol{w}_t)^\top \boldsymbol{H}(\boldsymbol{w}_t)(\boldsymbol{w}_{t+1} - \boldsymbol{w}_t).
\end{align}

Taking expectation over $\mathcal{B}$ and $\boldsymbol{z}$, we have \begin{align}\label{eq:pf1}
    \mathbb{E}[\mathcal{L}(\boldsymbol{w}_{t+1})] \leq \mathcal{L}(\boldsymbol{w}_t) - \eta\mathbb{E}_{\boldsymbol{z}, \mathcal{B}} \left[ \frac{\boldsymbol{z}^\top\nabla \mathcal{L}(\boldsymbol{w}_t)\boldsymbol{z}^\top \nabla{\mathcal{L}(\boldsymbol{w}_t, \mathcal{B})}}{|\boldsymbol{z}^\top \nabla{\mathcal{L}(\boldsymbol{w}_t, \mathcal{B})}|}\right] + \mathbb{E}_{\boldsymbol{z}, \mathcal{B}}\left[\frac{\eta^2 \boldsymbol{z}^\top \nabla{\mathcal{L}(\boldsymbol{w}_t, \mathcal{B})\boldsymbol{z}^\top\boldsymbol{H}(\boldsymbol{w}_t)\boldsymbol{z}^\top \nabla{\mathcal{L}(\boldsymbol{w}_t, \mathcal{B})\boldsymbol{z}}}}{2(\boldsymbol{z}^\top \nabla{\mathcal{L}(\boldsymbol{w}_t, \mathcal{B})})^2}\right].
\end{align}

For the negative term in \cref{eq:pf1}, the best case is when $\boldsymbol{z}^\top \nabla\mathcal{L}(\boldsymbol{w}_t, \mathcal{B})$ always holding the same sign to $\boldsymbol{z}^\top \nabla\mathcal{L}(\boldsymbol{w}_t)$. In this case, the term inside the expectation wrapper will always be positive to imply a descending loss. However, batch sampling will result in errors with \Cref{ass:srp}, which gives \begin{align}\label{eq:pf2}
    \mathbb{E}[\mathcal{L}(\boldsymbol{w}_{t+1})] \leq \mathcal{L}(\boldsymbol{w}_t) - \eta (1 - 2p_t) \mathbb{E}_{\boldsymbol{z}}\left[\frac{\boldsymbol{z}^\top \nabla{\mathcal{L}(\boldsymbol{w}_t)}\boldsymbol{z}^\top \nabla\mathcal{L}(\boldsymbol{w}_t)}{|\boldsymbol{z}^\top \nabla\mathcal{L}(\boldsymbol{w}_t)|}\right] + \mathbb{E}_{\boldsymbol{z}}\left[\frac{\eta^2 \boldsymbol{z}^\top \boldsymbol{H}(\boldsymbol{w}_t)\boldsymbol{z}}{2}\right].
\end{align}

Notice that $(\boldsymbol{z}^\top \nabla\mathcal{L}(\boldsymbol{w}_t))^2 / |\boldsymbol{z}^\top \nabla\mathcal{L}(\boldsymbol{w}_t)|$ is essentially $|\boldsymbol{z}^\top \nabla\mathcal{L}(\boldsymbol{w}_t)|$. Since $\boldsymbol{z}$ is a Gaussian vector centered at $0$, this term will be half-Gaussian with a mean of $\sqrt{\frac{2}{\pi}}\|\nabla\mathcal{L}(\boldsymbol{w}_t)\|_2$. By \cref{ass:r-effective}, $\text{tr}(\boldsymbol{H}(\boldsymbol{w}_t)) \leq Lr$, the bound will finally be \begin{align}
    \mathbb{E}(\mathcal{L}(\boldsymbol{w}_{t+1})) \leq \mathcal{L}(\boldsymbol{w}_t) - \eta (1 - 2p_t) \sqrt{\frac{2}{\pi}} \|\nabla\mathcal{L}(\boldsymbol{w}_t)\|_2 + \frac{\eta^2 L r}{2}.
\end{align}

\end{proof}


\label{appd:proof to prop steps}

\subsection{Proof to \Cref{thm:1}}

For FedSGD, we have the following well-known result:

\begin{lemma}[Dimension-free Descent Lemma for FedSGD] Given $\mathcal{L}(\boldsymbol{w})$ is a $L$-smooth function and $\hat{\nabla} \mathcal{L}(\boldsymbol{w}, \mathcal{B})$ an unbiased gradient estimator, the expected per-step loss descent can be bounded as follows:
    \begin{align}
        \mathbb{E}\left[\mathcal{L}(\boldsymbol{w}_{t+1})\right] \leq & \mathcal{L}(\boldsymbol{w}_t) - \eta \|\nabla \mathcal{L}(\boldsymbol{w}_t)\|_2^2 + \frac{L \eta^2}{2} \mathbb{E}_{k, \mathcal{B}}\left[\|\nabla\mathcal{L}_k(\boldsymbol{w}_t, \mathcal{B})\|_2^2\right].
    \end{align}
    \label{lm:descent_fo}
\end{lemma}

This result follows combining the unbiasedness of the FO gradient estimator and \Cref{ass:l-smooth}.

For ZO-FedSGD, 

\begin{proof}
We have \begin{align}
    &\mathbb{E}\left[\mathcal{L}(\boldsymbol{w}_{t+1})\right] \notag \\
    \leq& \mathcal{L}(\boldsymbol{w}_t) - \eta \|\nabla \mathcal{L}(\boldsymbol{w}_t)\|_2^2 + \frac{L \zeta \eta^2}{2} \mathbb{E}_{k, \mathcal{B}}\left[\|\nabla \mathcal{L}_k(\boldsymbol{w}_t, \mathcal{B})\|_2^2\right] \\
    \leq& \mathcal{L}(\boldsymbol{w}_t) - \eta \|\nabla \mathcal{L}(\boldsymbol{w}_t)\|_2^2 + \frac{L \zeta \eta^2}{2} c_g(1 + c_h) \|\nabla \mathcal{L}(\boldsymbol{w}_t)\|_2^2 + \frac{L \zeta \sigma_g^2 \eta^2}{2 KB} \mathbb{V}\left[\nabla \mathcal{L}(\boldsymbol{w}_t)\right] + \frac{L \zeta c_g \sigma_h^2 \eta^2}{2} \\
    \leq& \mathcal{L}(\boldsymbol{w}_t) - \left(\eta - \frac{L \zeta \eta^2 c_g(1 + c_h)}{2}\right) \|\nabla \mathcal{L} (\boldsymbol{w}_t)\|_2^2 + \frac{L \zeta \sigma_g^2 \eta^2}{2 KB} \mathbb{V}[\nabla \mathcal{L}(\boldsymbol{w}_t)] + \frac{L \zeta c_g \sigma_h^2 \eta^2}{2} \\
    \leq& \mathcal{L}(\boldsymbol{w}_t) - A_2(\mathcal{L}(\boldsymbol{w}_t) - \mathcal{L}^*) + C_2,
\end{align}
with a sufficiently small $\eta$ satisfying \begin{align}
    0 < \eta < 2 / L \zeta c_g(1 + c_h),
\end{align}
where \begin{align}
A_2 &= \left(2 \delta \eta - L \zeta \delta \eta^2 c_g (1 + c_h) - \frac{L \zeta \alpha \sigma_g^2 \eta^2}{KB}\right), \\
C_2 &= \frac{L \zeta c_g \sigma_h^2 \eta^2}{2}.
\end{align}

Substract $\mathcal{L}^*$ on both sides, then apply \Cref{ass:PL}, we have \begin{align}
    \mathbb{E}[\mathcal{L}(\boldsymbol{w}_{t+1})] - \mathcal{L}^* &\leq \left(1 - A_2\right) \left(\mathcal{L}(\boldsymbol{w}_t) - \mathcal{L}^*\right) + C_2. 
\end{align}

With proper redistribution of the $C_1$ term, we have an error bound $\tilde{C}_2 = C_2 / A_2$, and to reach an optimality gap smaller than $\epsilon$ will take \begin{align}
    t = A_2 \log \frac{\mathcal{L}(\boldsymbol{w}_0) - \mathcal{L}^* - \tilde{C}_2}{\epsilon}
\end{align} steps with $0 < A_2 < 1$.

We will have \begin{align}
    \mathbb{E}[\mathcal{L}(\boldsymbol{w}_{t+1})] - \mathcal{L}^* &\leq \left(1 - {A_1}\right) \left(\mathcal{L}(\boldsymbol{w}_t) - \mathcal{L}^*\right) + {C_1}. 
\end{align}
with similar processing for FedSGD for its exponential convergence, where \begin{align}
    A_1 &= \left(2 \delta \eta - L \delta \eta^2 c_g (1 + c_h) - \frac{L \alpha \sigma_g^2 \eta^2}{KB}\right), \\
    C_1 &= \frac{L c_g \sigma_h^2 \eta^2}{2}.
\end{align}


The quadratic term of weight difference vanishes since \textit{FeedSign} does not contain ``amplitude'' of the gradient projection.
\begin{align}
    \mathbb{E}[\mathcal{L}(\boldsymbol{w}_{t+1})] \leq \mathcal{L}(\boldsymbol{w}_t) - \eta (1 - 2p_t) \sqrt{\frac{2}{\pi}} \|\nabla\mathcal{L}(\boldsymbol{w}_t)\|_2 + \frac{\eta^2 L r}{2}.
\end{align}
Take $0 < \eta \leq \frac{1}{\|\nabla \mathcal{L}(\boldsymbol{w}_t)\|}$ to obtain a quadratic norm term, then apply \Cref{ass:PL}, then we have \begin{align}
    \mathbb{E}[\mathcal{L}(\boldsymbol{w}_{t+1})] \leq \mathcal{L}(\boldsymbol{w}_t) - \eta^2 (1 - 2p_t) \sqrt{\frac{2}{\pi}} 2 \delta (\mathcal{L}(\boldsymbol{w}_t) - \mathcal{L}^*) + \frac{\eta^2 L r}{2}.
\end{align}
Similar processes yields \begin{align}
    A_3& = 2\sqrt{\frac{2}{\pi}}\delta\eta^2(1 - 2 \max_{t} p_t), \\
    C_3& = \frac{L r \eta^2}{2}.
\end{align}

\end{proof}

\subsection{Proof to Proposition \ref{thm:reverse_prob}}

\begin{proof}
    In \emph{FeedSign}, assume at a particular point $\boldsymbol{w}_t$, the sign of the true gradient is $f_t = \nabla \mathcal{L}(\boldsymbol{w}_t) / |\mathcal{L}(\boldsymbol{w}_t)|$. We say that a client \textbf{successes} if it sends a correct sign to the PS, and \textbf{fails} otherwise. After local computation, honest clients always send the sign, and Byzantine clients always reverse the sign and then send it. Due to batch gradient noise, the probability of an honest fail is $p_{t, e}$ and an honest success is $1 - p_{t, e}$. Contradictorily, the probability of a Byzantine fail and Byzantine success is $1 - p_{t, e}$ and $p_{t, e}$, respectively. Assume the probability of a client being Byzantine is $p_{t, b}$.

    During a vote, the number of fails is a random variable $V$ that follows a binomial distribution with \begin{align}
        \mathbb{E}[V] &= \frac12 K + (\frac12 - p_{t, e})(2p_{t, b} - 1)K, \\
        \mathbb{V}[V] &= (\frac{1}{4} - p_{t, e}^2).
    \end{align}

    The adjusted error rate with Byzantine clients will be $\mathbb{E}[V] / K$.
\end{proof}

\subsection{Proof to \Cref{thm:dp}}

\begin{proof}
Denote $\mathcal{F} := \{1, -1\}$, and $\boldsymbol{p} := (p_1, \dots, p_K)$. Denote $\|\cdot, \cdot\|_1$ the Hamming distance of two vectors. Then for any $\boldsymbol{p} \in \mathcal{S}^K$ with $\|\boldsymbol{p}, \boldsymbol{p}'\|_1 \leq 1$ and any $f \in \mathcal{F}$, denoting $\hat{f} := f_{\text{DP}}(\boldsymbol{p})$, $\hat{f}' := f_{\text{DP}}(\boldsymbol{p}')$, \begin{align}
    &\frac{\text{Prob}(\hat{f} = f)}{\text{Prob}(\hat{f}' = f)} \notag \\
    =& \frac{\exp(\epsilon q_{\hat{f}} / 4)}{\exp(\epsilon q_{\hat{f}'} / 4)}  \frac{\exp(\epsilon q_{\hat{f}'} / 4) + \exp(\epsilon q_{-\hat{f}'} / 4)}{\exp(\epsilon q_{\hat{f}} / 4) + \exp(\epsilon q_{-\hat{f}} / 4)} \notag \\
    =& \exp\left(\frac{\epsilon(q_{\hat{f}} - q_{\hat{f}'})}{4}\right) \frac{\exp(\epsilon (q_{\hat{f}} + 2) / 4) + \exp(\epsilon (q_{-\hat{f}} + 2) / 4)}{\exp(\epsilon q_{\hat{f}} / 4) + \exp(\epsilon q_{-\hat{f}} / 4)} \notag \\
    \leq& \exp\left(\frac{2 \epsilon}{4}\right) \exp\left(\frac{2 \epsilon}{4}\right) \frac{\exp(\epsilon q_{\hat{f}} / 4) + \exp(\epsilon q_{-\hat{f}} / 4)} {\exp(\epsilon q_{\hat{f}} / 4) + \exp(\epsilon q_{-\hat{f}} / 4)} \notag \\
    =& \exp(\epsilon).
\end{align}
\end{proof}

\section{Test Accuracy of Other ZO Methods on CIFAR-10 Dataset}
\label{appd:zo acc cifar-10}

\begin{figure}
    \centering
    \includegraphics[width=0.4\textwidth]{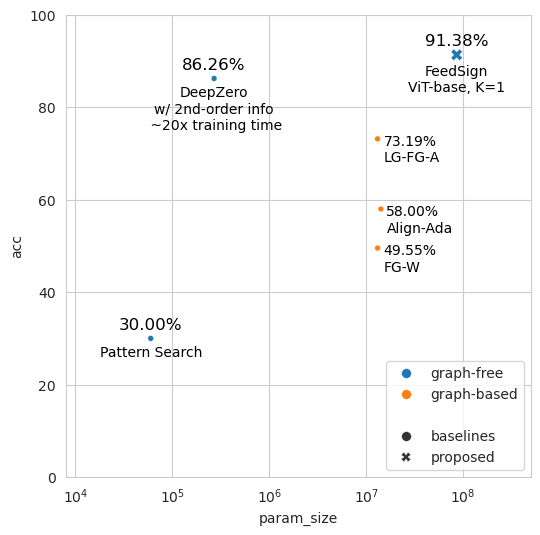}
    \caption{Test accuracy on CIFAR-10 dataset of some ZO baselines.\label{fig:CIFAR_ZO}}
\end{figure}

We list the test accuracy on the CIFAR-10 dataset obtained by some previous ZO from-the-scratch methods in \Cref{fig:CIFAR_ZO}.
 
These methods can be further split into two categories, \textit{computation graph-based} and \textit{computation graph-free} methods. Computation graph-based methods \cite{boopathy2022train, ren2022scaling} require information on the architecture of the neural networks, and the performance of such algorithms can be highly dependent on the structure of the DL models. The latter category \cite{chiang2022loss, chen2023deepzero} of algorithms relies solely on a finite number of objective values and is free from extra memory or computation cost for computational graphs but at the price of slower convergence. 
 
To conduct a fair comparison, we present the test accuracy when \emph{FeedSign} is run with only one client to simulate a centralized training manner since the listed baseline performances are all obtained under a centralized learning setting.

Notably, we are not able to compare our approach with DeepZero at the same scale. This is because in DeepZero, the authors used ResNeta-20, a version of ResNet tailored for images of small sizes ($3 \times 32 \times 32$, the standard CIFAR-10 size). Unfortunately, there are no available ``pre-trained models'' for models at this scale. In our implementation, we upsample the CIFAR-10 images to $3 \times 224 \times 224$ to adapt to the standard input shape of the standard version of ResNet-18. Moreover, DeepZero cannot scale up to ViT-large models due to prohibitively high computation overhead.

\section{Hyperparameters}

\begin{table*}[t]
\begin{center}
\caption{Memory consumption of OPT-1.3B with the MultiRC task. The reported memory does not include the cost of storing the model on the GPU. Batch size is set to $1$ to minimize the size of memory occupied by activations.}\label{table:memory_prof}
\begin{tabular}{lccc}
\toprule
\multirow{2}{*}{Task} & \emph{FeedSign}, Approach 2   & \emph{FeedSign}, Approach 1 & FO methods (common \textit{FedAvg})         \\ 
                      & Inference                             & Inference+optimizer  & Backpropagation \\ \midrule
Excess Memory (MB)    & $4027$                                & $10222$               & $46583$       \\ \bottomrule
\end{tabular}
\end{center}
\vspace{-15pt}
\end{table*}

We report the hyperparameters used in our experiments in \Cref{table:hyper}.

\begin{table*}[t]
\caption{Hyperparameters settings. *: See \Cref{table:step_and_perturbation} for details.}
\label{table:hyper}
\centering
\begin{tabular}{ccccccc}
\toprule
\multicolumn{1}{c}{}     & $B$                   & $T$                              & $\eta$             & $K$                   & $\mu$                      & $\beta$ \\ \hline
\multirow{2}{*}{\Cref{table:nlp_roberta_results}}                  &           \multirow{2}{*}{$64$}            &                  \multirow{2}{*}{$1 \times 10^{5}$}                &          $1 \times 10 ^ {-6}$ for \emph{ZO-FedSGD},          &            \multirow{2}{*}{$5$}           &              \multirow{2}{*}{$1 \times 10^{-3}$}              &            \multirow{2}{*}{-}                    \\
& & & $5 \times 10^{-5}$ for \emph{FeedSign} & & & \\ \hline
\multirow{2}{*}{\Cref{table:nlp_opt_results}}                  &          \multirow{2}{*}{$16$}             &                     \multirow{2}{*}{$2 \times 10^{4}$}             &           $1 \times 10 ^ {-7}$ for \emph{ZO-FedSGD},         &              \multirow{2}{*}{$5$}         &              \multirow{2}{*}{$1 \times 10^{-3}$}              &              \multirow{2}{*}{-}                  \\
 & &   &  $5 \times 10^{-6}$ for \emph{FeedSign}  & &  &  \\ \hline
\multirow{2}{*}{\Cref{table:vision models}}                  &               \multirow{2}{*}{$16$}        &             $2 \times 10^4$ for CIFAR-10                      &          \multirow{2}{*}{$1 \times 10^{-3}$}          &              \multirow{2}{*}{$5$}          &             \multirow{2}{*}{$1 \times 10^{-5}$}                &      \multirow{2}{*}{-}                           \\ 
&& $6 \times 10^4$ for CIFAR-100 & &&& \\ \hline
\multirow{2}{*}{\Cref{table:vision models byzantine}, \Cref{fig:bzt}} & \multirow{2}{*}{$64$} & $2 \times 10^4$ for CIFAR-10 & \multirow{2}{*}{$1 \times 10^{-3}$} & \multirow{2}{*}{$5$} & \multirow{2}{*}{$1 \times 10^{-5}$} & \multirow{2}{*}{-} \\ 
&& $6 \times 10^4$ for CIFAR-100 &&&& \\ \hline
\multirow{2}{*}{\Cref{table:differentK}}                  &       \multirow{2}{*}{$16$}                &         \multirow{2}{*}{*}                         &          $1 \times 10 ^ {-7}$ for \emph{ZO-FedSGD},          &            \multirow{2}{*}{*}           &                 \multirow{2}{*}{$1 \times 10^{-3}$}           &                     \multirow{2}{*}{-}           \\ 
&   &    &  $5 \times 10^{-6}$ for \emph{FeedSign}  &   &   &\\ \hline
\multirow{2}{*}{\Cref{table:non_iid}} & \multirow{2}{*}{$16$} & \multirow{2}{*}{$6 \times 10^4$} & $1 \times 10 ^ {-7}$ for \emph{ZO-FedSGD},       & \multirow{2}{*}{$5$} & \multirow{2}{*}{$1 \times 10^{-3}$} & \multirow{2}{*}{$1.0$}         \\ 
                         &                       &                                  & $5 \times 10^{-6}$ for \emph{FeedSign} &                       &                            &                                \\ \hline
\multirow{2}{*}{\Cref{table:byzantine_robustness}}                  &          \multirow{2}{*}{$16$}             &                 \multirow{2}{*}{$6 \times 10^4$}                 &          $1 \times 10 ^ {-7}$ for \emph{ZO-FedSGD},          &             \multirow{2}{*}{$5$}          &             \multirow{2}{*}{$1 \times 10^{-3}$}               &          \multirow{2}{*}{-}                      \\ 
&  &   &  $5 \times 10^{-6}$ for \emph{FeedSign} &    &   &    \\ \hline
\Cref{fig:noniid} & $64$ & $1.2 \times 10^{5}$ & $1 \times 10^{-4}$ & $25$ & $1 \times 10^{-5}$ & $1.0$ \\
\bottomrule
\end{tabular}

\end{table*}

\begin{table}[h]
\caption{Step budgets and numbers of perturbations used in \Cref{table:differentK}.}
\label{table:step_and_perturbation}
\centering
\begin{tabular}{ccc}
     \toprule
     $K$ & Step budget & The number of perturbations $T$ \\
     \hline
     \emph{MeZO} & $6 \times 10^4$ & $6 \times 10^4$ \\
     \hline
     $5$ & $6 \times 10^4$ & $3 \times 10^5$ \\
     $25$ & $1.2 \times 10^4$ & $3 \times 10^5$ \\
     \bottomrule
\end{tabular}
\end{table}

In \Cref{table:nlp_roberta_results} and \Cref{table:nlp_opt_results}, we kept the number of perturbations rather than step budget aligned to \emph{MeZO}. 
In \Cref{table:differentK}, we align the step budgets for $K=5$ in our comparison with the centralized counterpart. However, since the computation complexity scales to the number of perturbations hence to client pool size $K$ also, we report the result of $K=25$ with $1/5$ of step budgets. Other hyperparameters in \Cref{table:nlp_roberta_results}-\ref{table:differentK} is set to be consistent to \emph{MeZO} \cite{malladi2023fine}. 
Other hyperparameters in \Cref{table:nlp_roberta_results} are set to be consistent to \emph{MeZO} \cite{malladi2023fine}. 
In language model experiments, we set a larger learning rate $\eta$ since \emph{FeedSign} yields a smaller gradient norm since the amplitude is discarded. We believe that this will be partially accounted for outperforming the reported figures in \emph{MeZO} in several instances.

Additionally, we added a random multiplier following $1 + \mathcal{N}(0, 1)$ to gradient projection estimates of both \emph{ZO-FedSGD} and \emph{FeedSign} to simulate a high data heterogeneity with a high value of $c_g$ in \Cref{thm:1} in \Cref{fig:noniid}, apart from higher $\sigma_h$ caused by Dirichlet distributed client dataset.

\color{black} %

\section{Implementation Details}

\subsection{Experimental Settings}

To ensure consistency with previous research, we run the evaluation on RoBERTa-large, OPT-125M, and OPT-13B as is done in \emph{MeZO}. Additionally, we adapt the method to image models and run evaluations on ViT-base. Language models are run on a PS equipped with $8$ NVIDIA A100-80GB GPUs, and image models are run on a smaller PS equipped with $6$ NVIDIA GeForce RTX 3090 GPUs.

For hyperparameters, we follow the configuration of \emph{MeZO} for language models and develop our own set of parameters for image models. The number of participating clients is set to $K=5$. We set the random seed to $t$ at $t$-th step in \emph{FeedSign}.

In language models, the approach uses prompts that ensure the objective is close to that of the pertaining in finetuning, guaranteeing its good performance. In image models, the most straightforward way to adapt a pre-trained model to a new dataset with different number of classes is to change the size of the classifier layer.

\subsection{Model Parameter Update Using ZO Methods}
\label{appd:memory}

Two approaches can be used to update the model parameters in PyTorch: \begin{enumerate}
    \item Put the SPSA gradient estimate to the corresponding \verb|param.grad|, and use the standard PyTorch \verb|optimizer.step()| to update the model parameters.
    \item Inplace subtracting the entries of the \verb|state_dict| object of the PyTorch \verb|model| by the SPSA gradient estimate.
\end{enumerate}

In terms of memory requirement, our methods is essentially same to what is done in \cite{malladi2023fine}. The benefit of the first approach is that it is compatible with optimizers like Adam and RMSProp provided by PyTorch, and the drawback is that those optimizers often use momentum, resulting in $2$x or $3$x times the memory consumption compared to model inference, but still significantly smaller than that of the memory consumption of FO methods. The second approach consumes the exact same amount of memory compared to inference, however, it can result in slower convergence. We use the first approach for image models and the second for language models. We present a memory profiling in \Cref{table:memory_prof}. 

\subsection{Implementing \textit{FedSGD} over LLM}

A common way of building an FL system simulation is maintaining $K+1$ model instances in the memory ($K$ as the clients, one as the PS). However, the largest experiment we have carried out is FO full-parameter fine-tuning OPT-13B with \textit{FedAvg} over $K=5$ clients. Fine-tuning an instance of an OPT-13B model costs $316$ GB ($4$xA100 GPU) GPU memory, and $6$ of them is unaffordable for us. So we make use of the behavior of the auto-differentiation of PyTorch to detour the problem.

In a nutshell, we simulate only the global model on a virtual PS in the memory, and the local updates from different clients are accumulated to \verb|param.grad| by calling \verb|loss.backward()| for $K$ times, each time on the \verb|loss| computed from a corresponding client. The \verb|loss.backward()| implicitly simulates three steps in a communication round: 1) clients computing local updates, 2) clients sending local updates to the PS, 3) PS aggregates the local gradients. Next, a call of \verb|optimizer.step()| subtracts the parameters by the corresponding entries of \verb|param.grad|, simulating the global model marching a step along the gradient direction and broadcasting the updated model. In this way, we simulate the \textit{FedSGD} over $K$ clients using only $1$x inference memory.

We illustrate the process using the following snippet.

\vspace{15pt}

\begin{lstlisting}[language=Python]
for t in range(T):
    optimizer.zero_grad()
    for c in range(C):
        '''
        sample a batch from its private dataset
        calculate local loss
        '''
        loss.backward()     # accumulate local gradients to global model
    optimizer.step()
\end{lstlisting}

\section{When Will \emph{FeedSign} be Unbiased?}

It is clear that \emph{FeedSign} does not provide general unbiased gradient estimation. We will discuss a specific case where \emph{FeedSign} is unbiased.

Assume the true gradient projection for a specific model is $p$, and the estimation yielded by \emph{ZO-FedSGD} is $p_1$ and \emph{FeedSign} $p_2$. Consider a noise $n$ with CDF $F(x)$ on $p$ due to batch gradient estimation, specifically, $p_1 = p + n$, then \begin{align}
    \text{Prob}(p_2 = 1) &= F(p), \\
    \text{Prob}(p_2 = -1) &= F(-p).
\end{align}

To make $p_2$ an unbiased estimator of $p$, we need \begin{align}
    \mathbb{E}[p_2] = F(p) - F(-p) = p
\end{align}
for any $p$ in the support of $F(x)$. This result implies that \emph{FeedSign} is unbiased only when \begin{enumerate}
    \item the noise CDF is uniform on $[-1, 1]$,
    \item $p$ takes value on $[-1, 1]$ only.
\end{enumerate}

This is an obviously unrealistic setting. However, there could be methods that distort the true distribution of $p$ and $n$ so that they behave similarly to the abovementioned case, making \emph{FeedSign} an unbiased FL method. This could be a topic for future investigation.

\color{black}

\section{Future Research Directions}

This work opens new topics for future investigations, including:

\textbf{Convergence acceleration by seed selection.} As pointed out in \Cref{thm:1} and Remark \ref{rmk:crc}, descent directions are not equal. However, the methods under investigation employ an \emph{isotropic} exploration tactic, that is, indiscriminately accepting an assigned gradient direction. Although it performs well in most cases, it contradicts the \emph{anisotropic} \cref{ass:r-effective}, and may perform worse under certain settings as \Cref{table:nlp_roberta_results_512}, suggesting that there may be potential for improvement in convergence.

\begin{table*}[h]
    \setlength{\tabcolsep}{5pt}
    \caption{Results on RoBERTa-large over language tasks. The best results obtained using federated ZO optimization is \textbf{bolded}, and the metric gap to that of the FO method is reported in the rightmost column.}
    \label{table:nlp_roberta_results_512}
    \begin{center}
        \begin{tabular}{lccccccc}
            \toprule
            \multicolumn{1}{c}{Task}  &\multicolumn{1}{c}{\bf SST-2}  &\multicolumn{1}{c}{\bf SST-5}  &\multicolumn{1}{c}{\bf SNLI}  &\multicolumn{1}{c}{\bf MNLI} &\multicolumn{1}{c}{\bf RTE}  &\multicolumn{1}{c}{\bf TREC} & \multirow{2}*{Gap} \\ 
            \multicolumn{1}{c}{Type} & \multicolumn{2}{c}{---- sentiment ----} & \multicolumn{3}{c}{- natural language inference -} & \multicolumn{1}{c}{-- topic --} & \\
            \hline
            Zero-shot & 79.0 & 35.5 & 50.2 & 48.8 & 51.4 & 32.0 & -- \\
            \hline
            \multicolumn{7}{c}{$\text{number of shots}=512$} \\
            \hdashline
            FO & 93.9 & 55.9 & 88.7 & 84.4 & 82.7 & 97.3 & -- \\
            \hdashline
            MeZO & 93.3 & 53.2 & 83.0 & 78.3 & 78.6 & 94.3 & -3.7 \\
            \hdashline
            \emph{ZO-FedSGD} & \bf 93.3&	\bf 52.0&	\bf 84.1&	\bf 76.7&	\bf 77.8 &	\bf 94.3& \bf -4.0 \\
            & (0.4)&	(1.3)&	(0.6)&	(1.8)&	(1.1) &	(1.1)& -- \\
            \bf \emph{FeedSign} & 92.6&	50.1&	83.0&	75.6&	76.3&	93.1& -5.3 \\
            & (0.7)&	(0.3)&	(0.5)&	(0.3)&	(3.5)&	(0.6)& -- \\
            \bottomrule
        \end{tabular}
    \end{center}
\end{table*}

\textbf{Differentially private variants.} Although this work proposed a differentially private version of \emph{FeedSign}, the convergence-privacy trade-off is not clearly characterized. Moreover, there are other metrics of privacy characterization, and there could be an underexplored complicated interplay between convergence and privacy.

\textbf{Fine-tuning acceleration by inference acceleration.} For similar designs that use backward passes-free fine-tuning, proper forward pass acceleration techniques can be integrated for a faster FFT process in real-world time.




\end{document}